\documentclass[11pt]{article}

\usepackage{dsfont}

\def\colorful{1}

\usepackage{amsfonts,amssymb,epsfig,color,float,graphicx,verbatim, enumitem}
\usepackage{multirow}
\usepackage{algorithm}
\usepackage[noend]{algpseudocode}

\usepackage{hyperref}
\hypersetup{
    colorlinks = true,
    citecolor = {blue}
}

\usepackage{enumitem}

\usepackage{framed}
\usepackage{nicefrac}

\usepackage{graphicx}
\usepackage{amssymb}
\usepackage{xcolor}
\usepackage{enumitem}
\usepackage{mathtools}
\usepackage{latexsym}
\usepackage{amsmath}
\usepackage{amsfonts}
\usepackage{amsmath}
\usepackage{amsthm}
\usepackage{bbm}
\usepackage{centernot}
\usepackage{booktabs}

\def\nnewcolor{1}
\ifnum\nnewcolor=1

\fi
\ifnum\nnewcolor=0

\fi

\ifnum\colorful=1

\else

\fi

\newtheorem{theorem}{Theorem}[section]

\newtheorem{lemma}[theorem]{Lemma}
\newtheorem{informal theorem}[theorem]{Theorem (informal statement)}

\newtheorem{proposition}[theorem]{Proposition}
\newtheorem{corollary}[theorem]{Corollary}
\newtheorem{claim}[theorem]{Claim}
\newtheorem{fact}[theorem]{Fact}

\theoremstyle{definition}
\newtheorem{definition}[theorem]{Definition}
\newtheorem{remark}[theorem]{Remark}
\newtheorem{problem}[theorem]{Problem}

\newcommand{\eqdef}{\stackrel{{\mathrm {\footnotesize def}}}{=}}

\makeatletter
  \def\cstar#1{\expandafter\@cstar\csname c@#1\endcsname}
  \def\@cstar#1{\ifcase#1\or $\ast$\or $\ast\ast$\or $\ast\ast\ast$\fi}
  \AddEnumerateCounter{\cstar}{\@cstar}{$\ast\ast\ast$}
\makeatother

\newcommand{\p}{\mathbf{P}}

\newcommand{\R}{\mathbb{R}}

\newcommand{\Z}{\mathbb{Z}}
\newcommand{\N}{\mathbb{N}}
\newcommand{\eps}{\epsilon}

\newcommand{\poly}{\mathrm{poly}}
\newcommand{\var}{\mathbf{Var}}

\newcommand{\dotp}[2]{\langle #1, #2 \rangle}

\newcommand{\cA}{\mathcal{A}}

\newcommand{\cC}{\mathcal{C}}
\newcommand{\cD}{\mathcal{D}}
\newcommand{\cE}{\mathcal{E}}
\newcommand{\cF}{\mathcal{F}}

\newcommand{\cH}{\mathcal{H}}

\newcommand{\cL}{\mathcal{L}}

\newcommand{\cN}{\mathcal{N}}

\newcommand{\cP}{\mathcal{P}}
\newcommand{\cQ}{\mathcal{Q}}

\newcommand{\cS}{\mathcal{S}}

\newcommand{\cZ}{\mathcal{Z}}

\newcommand{\dx}{\mathrm{dx}}
\newcommand{\dy}{\mathrm{dy}}

\allowdisplaybreaks
\newcommand{\snorm}[2]{\left\| #2 \right\|_{#1}}
\DeclareMathOperator*{\E}{\mathbf{E}}
\DeclareMathOperator*{\pr}{\mathbf{Pr}}
\DeclarePairedDelimiter\abs{\lvert}{\rvert}

\title{Statistical Query Lower Bounds for \\ 
List-Decodable Linear Regression}

\author{
Ilias Diakonikolas\thanks{Supported by NSF Award CCF-1652862 (CAREER), a Sloan Research Fellowship, and 
a DARPA Learning with Less Labels (LwLL) grant.}\\
University of Wisconsin-Madison\\
{\tt ilias@cs.wisc.edu}\\
\and
Daniel M. Kane\thanks{Supported by NSF Award CCF-1553288 (CAREER) and a Sloan Research Fellowship.}\\
University of California, San Diego\\
{\tt dakane@cs.ucsd.edu}
\and
Ankit Pensia\thanks{Supported by NSF Award DMS-1749857.}\\
University of Wisconsin-Madison\\
{\tt ankitp@cs.wisc.edu}\\
\and
Thanasis Pittas\thanks{Supported in part by NSF Award DMS-2023239 (TRIPODS).}\\
University of Wisconsin-Madison\\
{\tt pittas@wisc.edu}\\
\and
Alistair Stewart\\
Web 3 Foundation\\
{\tt stewart.al@gmail.com}\\
}

\begin{document}

\maketitle

\begin{abstract}
We study the problem of list-decodable linear regression, where an adversary
can corrupt a majority of the examples. Specifically, we are given a set $T$ of labeled examples
$(x, y) \in \R^d \times \R$ and a parameter $0< \alpha <1/2$ such that an $\alpha$-fraction
of the points in $T$ are i.i.d.\ samples from a linear regression model with Gaussian covariates,
and the remaining $(1-\alpha)$-fraction of the points are drawn from an arbitrary noise distribution.
The goal is to output a small list of hypothesis vectors such that at least one of them is close to the target regression vector. 
Our main result is a Statistical Query (SQ) lower bound of $d^{\poly(1/\alpha)}$ for this problem.
Our SQ lower bound qualitatively matches the performance of previously developed algorithms,
providing evidence that current upper bounds for this task are nearly best possible.
\end{abstract}

\setcounter{page}{0}

\thispagestyle{empty}

\newpage

\section{Introduction} \label{sec:intro}

\subsection{Background and Motivation} \label{ssec:background}

Linear regression is one of the oldest and most fundamental statistical tasks
with numerous applications in the sciences~\cite{Rousseeuw:1987, Dielman01, McD09}.
In the standard setup, the data are labeled examples $(x^{(i)}, y^{(i)})$,  
where the examples (covariates) $x^{(i)}$ are i.i.d.\ samples from a distribution $D_x$ on $\R^d$ 
and the labels $y^{(i)}$ are noisy evaluations of a linear function. More specifically, each label 
is of the form $y^{(i)} = \beta \cdot x^{(i)} + \eta^{(i)}$, 
where $\eta^{(i)}$ is the observation noise, for an unknown target regression vector $\beta \in \R^d$. 
The objective is to approximately recover the hidden regression vector. 
In this basic setting, linear regression is well-understood. For example,
under Gaussian distribution, the least-squares estimator is known to be statistically  
and computationally efficient. %

Unfortunately, classical efficient estimators inherently fail in the presence of even a very small 
fraction of adversarially corrupted data. In several applications of modern data analysis, 
including machine learning security~\cite{Barreno2010,BiggioNL12, SteinhardtKL17, DKK+19-sever} 
and exploratory data analysis, e.g., in biology~\cite{RP-Gen02, Pas-MG10, Li-Science08},
typical datasets contain arbitrary or adversarial outliers. Hence, 
it is important to understand the algorithmic possibilities and fundamental limits 
of learning and inference in such settings. 
Robust statistics focuses on designing estimators tolerant to a small amount
of contamination, where the outliers are the {\em minority} of the dataset. 
Classical work in this field~\cite{HampelEtalBook86, Huber09} developed 
robust estimators for various basic tasks, alas with exponential runtime. 
More recently, a line of work in computer science, 
starting with~\cite{DKKLMS16, LaiRV16}, developed the first computationally efficient robust 
learning algorithms for various high-dimensional tasks. Subsequently, there has been 
significant progress in algorithmic robust statistics by several communities, 
see~\cite{DK19-survey} for a survey on the topic.

In this paper, we study high-dimensional robust linear regression 
in the presence of a {\em majority} of adversarial outliers.
As we explain below, in several applications, 
asking for a minority of outliers is too strong of an assumption. 
It is thus natural to ask what notion of learning can capture the regime 
when the clean data points (inliers) constitute the {\em minority} of the dataset. 
While outputting a {\em single} accurate hypothesis in this regime
is information-theoretically impossible, one may be able to compute a {\em small list}
of hypotheses with the guarantee that {\em at least one of them} is accurate.
This relaxed notion is known as \emph{list-decodable learning}~\cite{BalcanBV08, CSV17},
formally defined below.

\begin{definition}[List-Decodable Learning]\label{def:ld}
Given a parameter $0<\alpha < 1/2$ and a distribution family $\cD$ on $\R^d$, 
the algorithm specifies $n \in \Z_+$ and observes $n$ i.i.d.\ samples from a distribution 
$E = \alpha D + (1{-\alpha}) N$, where $D$ is an unknown distribution in $\cD$ 
and $N$ is arbitrary. We say $D$ is the distribution of inliers, 
$N$ is the distribution of outliers, and $E$ is an $(1{-\alpha})$-corrupted version of $D$. 
Given sample access to an $(1{-\alpha})$-corrupted version of $D$, the goal 
is to output a ``small'' list of hypotheses {$\cL$} at least one of which is (with high probability) close to the target parameter of $D$.
\end{definition}

We note that a list of size $O(1/\alpha)$ typically suffices; an algorithm with a 
$\poly(1/\alpha)$ sized list, or even a worse function of $1/\alpha$ (but independent of the dimension $d$) 
is also considered acceptable.

Natural applications of list-decodable learning include crowdsourcing, 
where a majority of participants could be unreliable~\cite{SVC16-nips,MeisterV18}, 
and semi-random community detection in stochastic block models~\cite{CSV17}. 
List-decoding is also useful in the context of semi-verified learning~\cite{CSV17,MeisterV18}, 
where a learner can audit a very small amount of trusted data.
If the trusted dataset is too small to directly learn from, using a list-decodable learning procedure, 
one can pinpoint a candidate hypothesis consistent with the verified data. 
Importantly, list-decodable learning generalizes
the task of learning mixture models, see, 
e.g.,~\cite{DV89, JJ94, ZJD16, LL18, KC19, CLS20, DK20-ag} 
for the case of linear regression studied here. 
Roughly speaking, by running a list-decodable estimation procedure 
with the parameter $\alpha$ equal to the smallest mixing weight, each true cluster of points 
is an equally valid ground-truth distribution, so the output list must contain candidate parameters 
close to each of the true parameters. 

In list-decodable linear regression (the focus of this paper),
$D$ is a distribution on pairs $(X,y)$, where $X$ is a standard Gaussian on $\R^d$, 
$y$ is approximately a linear function of $x$, and the algorithm 
is asked to approximate the hidden regressor.
The following definition specifies the distribution family $\cD$ of the inliers 
for the case of linear regression with Gaussian covariates. 

\begin{definition}[Gaussian Linear Regression]\label{def:glr} 
Fix $\sigma > 0$. For $\beta \in \R^d$, let $D_{\beta}$ be the distribution over 
$(X,y)$, $X\in \R^d$, $y \in \R$, such that $X \sim \cN(0,I_d)$ and $y = \beta^T X + \eta $, 
where $\eta \sim \cN(0,\sigma^2)$ independently of $X$. 
We define $\cD$ to be the set $\{D_{\beta}: \beta \in S'\}$ for some set $S' \subseteq \R^d$. 
\end{definition}

Recent algorithmic progress~\cite{KKK19-list, RY19-list} has been made on this problem 
using the SoS hierarchy. The guarantees in~\cite{KKK19-list, RY19-list} are very far from the
information-theoretic limit in terms of sample complexity. 
In particular, they require $d^{\poly(1/\alpha)}$ samples and time to obtain non-trivial error guarantees (see Table~\ref{tab:results}): 
\cite{KKK19-list} obtains an error guarantee of $O(\sigma/\alpha)$ 
with a list of size $O(1/\alpha)$, whereas \cite{RY19-list} obtains an error guarantee 
of $O(\sigma/\alpha^{3/2})$ with a list of size $(1/\alpha)^{O(\log(1/\alpha))}$.

On the other hand, as shown in this paper (see Theorem~\ref{thm:sc}), 
$\poly(d/\alpha)$ samples information-theoretically suffice to obtain near-optimal error guarantees. 
This raises the following natural question:
\begin{center}
{\em What is the complexity of list-decodable linear regression? \\ 
Are there efficient algorithms with significantly better sample-time tradeoffs?}
\end{center}

We study the above question in a natural and well-studied restricted model of computation,
known as the Statistical Query (SQ) model~\cite{Kearns:98}.
As the main result of this paper, we prove strong SQ lower bounds for this problem.
Via a recently established equivalence~\cite{brennan2020statistical}, 
our SQ lower bound also implies low-degree testing lower bounds for this task.
Our lower bounds can be viewed as evidence that current upper bounds for this problem may be
qualitatively best possible.

Before we state our contributions in detail, we give some background on
SQ algorithms. SQ algorithms are a broad class of algorithms
that  are only allowed to query expectations of bounded functions of the distribution
rather than directly access samples. 
Formally, an SQ algorithm has access to the following oracle.

\begin{definition}[STAT Oracle] \label{def:stat}
Let $D$ be a distribution on $\R^d$. A statistical query is a bounded function $q : \R^d \to [-1,1]$. 
For $\tau>0$, the $\mathrm{STAT}(\tau)$ oracle responds to the query $q$ with a value $v$ 
such that $|v - \E_{X \sim D}[q(X)] | \leq \tau$.
We call $\tau$ the tolerance of the statistical query.
\end{definition}

The SQ model was introduced by Kearns~\cite{Kearns:98} in the context of supervised learning
as a natural restriction of the PAC model~\cite{Valiant:84}. Subsequently, the SQ model
has been extensively studied in a plethora of contexts (see, e.g.,~\cite{Feldman16b} and references therein).
The class of SQ algorithms is rather broad and captures a range of known supervised learning algorithms. 
More broadly, several known algorithmic techniques in machine learning
are known to be implementable using SQs. These include spectral techniques,
moment and tensor methods, local search (e.g., Expectation Maximization),
and many others (see, e.g.,~\cite{FeldmanGRVX17, FeldmanGV17}).

\subsection{Our Results} \label{ssec:results}

We start by showing that $\poly(d/\alpha)$ samples are sufficient to obtain a near-optimal 
error estimator, albeit with a computationally inefficient algorithm.

\begin{theorem}[Information-Theoretic Bound]\label{thm:sc}
There is a (computationally inefficient) list-decoding algorithm for Gaussian linear regression 
that uses $O(d/\alpha^3)$ samples, returns a list of $O(1/\alpha)$ many hypothesis vectors, 
and has $\ell_2$-error guarantee of  $O((\sigma/\alpha)\sqrt{\log(1/\alpha)})$. 
Moreover, if the dimension $d$ is sufficiently large, any list-decoding algorithm 
that outputs a list of size $\poly(1/\alpha)$ must have $\ell_2$-error at least 
$\Omega((\sigma/\alpha)/\sqrt{\log(1/\alpha)})$.
\end{theorem}

The proof of this result is given in Section~\ref{sec:info-theoretic} 
(see Theorems~\ref{thm:sample-ub} and~\ref{thm:sample-lb}).
Our main result is a strong SQ lower bound for the list-decodable Gaussian linear regression problem.
We establish the following theorem (see Theorem~\ref{thm:main-formal} for a more detailed formal statement).

\begin{theorem}[SQ Lower Bound] \label{thm:sq-lb-inf}
Assume that the dimension $d \in \Z_+$ is sufficiently large and consider the problem of list-decodable linear regression, 
where  the fraction of inliers is $\alpha \in (0,1/2)$, the regression  vector $\beta \in \R^d$ has norm $\|\beta\|_2 \leq 1$, 
and the additive noise has standard deviation $\sigma \leq \alpha$. Then any SQ algorithm that returns a list 
${\cal L} $ of candidate vectors containing a $\widehat{\beta}$ 
such that $\|\widehat{\beta} - \beta \|_2 \leq 1/4$ does one of the following:
\begin{itemize}
    \item it uses at least one query with tolerance at most $d^{-\Omega(1/\sqrt{a})}/ \sigma$,
    \item it makes $2^{d^{\Omega(1)}}$ queries, or
    \item it returns a list of size $|{\cal L}|  = 2^{d^{\Omega(1)}}$.
\end{itemize}
\end{theorem}

\begin{table}
\centering
\begin{tabular}{@{}lllll@{}}
\toprule
Algorithmic Result & Sample Size & Running Time & List size
\\ \midrule
Karmalkar-Klivans-Kothari~
\cite{KKK19-list}     & $(d/\alpha)^{O\left(1/\alpha^4\right)}$        & $(d/\alpha)^{O\left(1/\alpha^8\right)}$      & $O(1/\alpha)$
\\
Raghavendra and Yau~
\cite{RY19-list}    & $d^{O\left(1/\alpha^4\right)}$  & $d^{O\left(1/\alpha^8\right)}(1/\alpha)^{\log(1/\alpha)}$         
&   $(1/\alpha)^{O\left(\log(1/\alpha)\right)}$
\\ \bottomrule
\end{tabular}
\caption{The table summarizes the sample complexity, running time, and list size of 
the known list-decodable linear regression algorithms in order to obtain a $1/4$-additive approximation to the hidden
regression vector $\beta$ in the setting of Theorem~\ref{thm:sq-lb-inf}, 
i.e., when $\|\beta\|_2 \leq 1$ and $\sigma$ is sufficiently small as a function of $\alpha$: 
\cite{KKK19-list} requires $\sigma=O(\alpha)$ and \cite{RY19-list} requires $\sigma=O(\alpha^{3/2})$.}
\label{tab:results}
\end{table}

Informally speaking, Theorem~\ref{thm:sq-lb-inf} shows that no SQ algorithm 
can approximate $\beta$ to constant accuracy with a sub-exponential in $d^{\Omega(1)}$ size list 
and sub-exponential in $d^{\Omega(1)}$ many queries, 
unless using queries of very small tolerance  -- that would require at least
$\sigma d^{\Omega(1/\sqrt{\alpha})}$ samples to simulate. For $\sigma$ not too small,
e.g., $\sigma = \poly(\alpha)$, in view of Theorem~\ref{thm:sc}, 
this result can be viewed as an information-computation tradeoff
for the problem, within the class of SQ algorithms.

A conceptual implication of Theorem~\ref{thm:sq-lb-inf} is that list-decodable linear regression 
is harder (within the class of SQ algorithms) than the related problem of learning mixtures 
of linear regressions (MLR).
Recent work~\cite{DK20-ag} gave an algorithm (easily implementable in SQ) 
for learning MLR with $k$ equal weight separated components (under Gaussian 
covariates) with sample complexity and running time $k^{\mathrm{polylog} (k)}$, i.e., 
{\em quasi-polynomial} in $k$. 
Recalling that one can reduce $k$-MLR (with well-separated components) 
to list-decodable linear regression for $\alpha = 1/k$, 
Theorem~\ref{thm:sq-lb-inf} implies that the aforementioned algorithmic result 
cannot be obtained via such a reduction.

\begin{remark} \label{rem:ld-test}
While the main focus of this work is on the SQ model, our result has immediate implications
to a related popular restricted computational model --- that of low-degree (polynomial) algorithms~\cite{HopkinsS17,HopkinsKPRSS17,Hopkins-thesis}. 
Recent work~\cite{brennan2020statistical} established that (under certain assumptions) 
an SQ lower bound also implies a qualitatively similar lower bound in the low-degree model. 
We leverage this connection to show a similar lower bound in this model 
(see Section~\ref{sec:hardness_again_low_degree_poly}).
\end{remark}

\subsection{Overview of Techniques} \label{ssec:techniques}

In this section, we provide a detailed overview of our SQ lower bound construction.
We recall that there exists a general methodology for establishing 
SQ lower bounds via an appropriate complexity measure, known as SQ dimension. 
Several related notions of SQ dimension exist in the literature, 
see, e.g.,~\cite{BFJ+:94, FeldmanGRVX17, Feldman17}. Here we focus on the 
framework introduced in~\cite{FeldmanGRVX17} for search problems over distributions, 
which is more natural in our setting. A lower bound on the SQ dimension of a search problem 
provides an unconditional lower bound on the SQ complexity of the problem. Roughly speaking,
for a notion of correlation between distributions in our family $\mathcal{D}$ 
(Definition~\ref{def:pc}), establishing an SQ lower bound amounts 
to constructing a large cardinality sub-family $\mathcal{D}' \subseteq \mathcal{D}$ 
such that every pair of distributions in  $\mathcal{D}'$ are nearly uncorrelated 
with respect to a given reference distribution $R$ 
(see Definition~\ref{def:sqdim} and Lemma~\ref{cor:cor3-12-FGR}).

A general framework for constructing SQ-hard families of distributions was introduced
in~\cite{DKS17-sq}, which showed the following: 
Let the reference distribution $R$ be $\cN(0,I)$
and $A$ be a univariate distribution whose low-degree moments match those
of the standard Gaussian (and which satisfies an additional mild technical condition).
Let $P_{A, v}$ be the distribution that is a copy of $A$ in the $v$-direction 
and standard Gaussian in the orthogonal complement (Definition~\ref{def:high-dim-distribution}).
Then the distribution family $\{ P_{A, v} \}_{v \in S}$, 
where $S$ is a set of nearly orthogonal unit vectors, satisfies the pairwise nearly uncorrelated
property (Lemma~\ref{lem:lemma-3-4-DKS17}), and is therefore SQ-hard to learn.

Unfortunately, the~\cite{DKS17-sq} framework does not suffice in the supervised setting of the current paper
for the following reason: The joint distribution over labeled examples $(X, y)$ in our setting  
does not possess the symmetry properties required for moment-matching 
with the reference $R=\cN(0,I)$ to be possible. Specifically, the behavior of $y$
will necessarily be somewhat different than the behavior of $X$. 
To circumvent this issue, we leverage an idea from~\cite{DKS19}. The high-level idea is to construct distributions 
$E_v$ on $(X, y)$ such that for any fixed value $y_0$ of $y$, 
the conditional distribution of $X \mid y = y_0$ under $E_v$
is of the form  $P_{A,v}$ described above, where $A$ is replaced with some $A_{y_0}$.

We further explain this modified construction. Note that $E_v$ should be of the form $\alpha D_v + (1{-\alpha}) N_v$, where $D_v$ is the inlier
distribution (corresponding to the clean samples from the linear regression model) and $N_v$ 
is the outlier (noise) distribution. To understand what properties our distribution should satisfy, 
we start by looking at the inlier distribution $D$. By definition, for $(X, y) \sim D$, we have that 
$y = \beta^T X + \eta$, where $X \sim N(0, I)$ and $\eta \sim N(0, \sigma^2)$ is independent of $X$.
A good place to start here is to understand the distribution of $X$ conditioned on $y=y_0$, 
for some $y_0$, under $D$. 
It is not hard to show (Fact~\ref{fact:inliers_conditional_dist}) 
that this conditional distribution is already of the desired form $P_{A,\beta}$: 
it is a product of a $(d-1)$-dimensional standard Gaussian in directions orthogonal to $\beta$, 
while in the $\beta$-direction it is a much narrower Gaussian with mean proportional to $y_0$.
To establish our SQ-hardness result, we would like to mix
this conditional distribution with a carefully selected outlier distribution $N \mid y = y_0$, 
such that the resulting mixture $E \mid y = y_0$ matches many of its low-degree moments with the standard Gaussian in the $\beta$-direction, while being standard Gaussian in the orthogonal directions.
In the setting of minority of outliers, \cite{DKS19} 
was able to provide an explicit formula for $N$ 
and match {\em three} moments to show an SQ lower bound of $\Omega(d^2)$.
The  main technical difficulty in our paper is that, in order to prove the desired SQ lower bound 
of $\Omega(d^{\poly(1/\alpha)})$, we need to match $\poly(1/\alpha)$ many moments. We explain how to achieve this below.

Here we take a different approach and establish the existence of the desired outlier distribution $N|y=y_0$ in a non-constructive manner. We note that our problem is an instance of the moment-matching problem, 
where given a sequence of real numbers, the goal is to decide whether a distribution exists 
having that sequence as its low-degree moments. At a high-level, we leverage classical results 
that tackle this general question by formulating a linear program (LP) and using LP-duality to derive necessary 
and sufficient feasibility conditions (see~\cite{karlin1953geometry} and Theorem~\ref{ThmDualityBdd}).
This moment-matching via LP duality approach is fairly general, but stumbles upon two technical obstacles
in our setting.

The first technical issue is that our final distributions $E_v$ on $(X, y)$ need to have bounded 
$\chi^2$-divergence with respect to the reference distribution, 
since the pairwise correlations scale with this quantity (see Lemma~\ref{lem:lemma-3-4-DKS17}).  
To guarantee this, we can ensure that the outlier distribution in the $\beta$-direction 
is in fact equal to the convolution of a distribution with bounded support
with a narrow Gaussian: (i) The contraction property of this convolution operator means 
that it can only reduce the $\chi^2$-divergence, and 
(ii) the bounded support can be used in combination with tail-bounds on Hermite polynomials (Lemma~\ref{PropTailBoundHermite2}) to bound from above the contribution to the $\chi^2$-divergence 
of each Hermite coefficient of our distribution (Lemma~\ref{lem:existence_of_A_y}). 
These additional constraints necessitate a modification to the moment-matching problem, 
but it can still be readily analyzed (Theorem~\ref{ThmHardDist_1}).

The second and more complicated issue involves the fraction of outliers, i.e., the parameter ``$1{-\alpha}$''. 
Unfortunately, it is easy to see that the fraction of outliers necessary to make the conditional distributions match the desired number of moments must necessarily go to $1$ as $|y|$ goes to infinity: 
As $|y|$ gets bigger, the conditional distribution of inliers moves further away from $\cN(0,I)$ (Fact~\ref{fact:inliers_conditional_dist}) and thus needs to be mixed more heavily with outliers to be corrected.
This is a significant problem, since by definition we can only afford to use a $(1{-\alpha})$-fraction of outliers overall.
To handle this issue, we consider a reference distribution $R$ on $(X, y)$ that has much heavier tails in $y$ 
than the distribution of inliers has. This essentially means that as $|y|$ gets large, 
the conditional probability that a sample is an outlier gets larger and larger. This is balanced by having 
slightly lower fraction of outliers for smaller values of $|y|$, 
in order to ensure that the total fraction of outliers is still at most $1{-\alpha}$. 
To address this issue, we leverage the fact that the probability that a clean sample 
has large value of $|y|$ is very small. Consequently, we can afford to make the error rates
for such $y$ quite large without increasing the overall probability 
of error by very much.

	\subsection{Preliminaries} \label{ssec:prelims}

\paragraph{Notation} We use $\N$ to denote natural numbers and $\Z_+$ to denote positive integers. For $n \in \Z_+$ we  denote $[n] \eqdef \{1,\ldots,n\}$ and use $\cS^{d-1}$ for the $d$-dimensional unit sphere. We denote by $\mathbf{1}(\cE)$ the indicator function of the event $\cE$. We use $I_d$ to denote the $d\times d$ identity matrix. For a random variable $X$, we use $\E[X]$ for its expectation.	
	For $m\in\Z_+$, the $m$-th moment of $X$ is defined as $\E [X^m]$.	
	 We use $\cN(\mu,\Sigma)$ to denote the Gaussian distribution with mean $\mu$ and covariance matrix $\Sigma$. We let $\phi$ denote the pdf of the one-dimensional standard Gaussian.	
	 When $D$ is a distribution, we use $X \sim D$ to denote that the random variable $X$ is distributed according to $D$.
	 For a vector $x \in \R^d$, we let $\snorm{2}{x}$ denote its $\ell_2$-norm. For $y\in \R$, we denote by $\delta_y$  the Dirac delta distribution at $y$, i.e., the distribution that assigns probability mass 1 to the single point $y \in \R$ and zero elsewhere.
	 When there is no confusion, we will use the same letters for distributions and their probability density functions.

	\paragraph{Hermite Analysis} 
Hermite polynomials form a complete orthogonal basis of the vector space $L^2(\R,\cN(0,1))$ of all functions $f:\R \to \R$ such that $\E_{X\sim \cN(0,1)}[f^2(X)]< \infty$. There are two commonly used types of Hermite polynomials. The \emph{physicist's } Hermite polynomials, denoted by $H_k$ for $k\in \N$ satisfy the following orthogonality property with respect to the weight function $e^{-x^2}$: for all $k,m \in \N$, $\int_\R H_k(x) H_m(x) e^{-x^2} \dx = \sqrt{\pi} 2^k k! \mathbf{1}(k=m)$. The \emph{probabilist's} Hermite polynomials $H_{e_k}$ for $k\in \N$ satisfy $\int_\R H_{e_k}(x) H_{e_m}(x) e^{-x^2/2} \dx = k! \sqrt{2\pi}  \mathbf{1}(k=m)$ and are related to the physicist's polynomials through $H_{e_k}(x)=2^{-k/2}H_k(x/\sqrt{2})$. 
	We will mostly use the \emph{normalized probabilist's} Hermite polynomials, $h_k(x) = H_{e_k}(x)/\sqrt{k!}$, $k\in \N$ for which $\int_\R h_k(x) h_{m}(x) e^{-x^2/2} \dx = \sqrt{2\pi} \mathbf{1}(k=m)$.
	These polynomials are the ones obtained by Gram-Schmidt orthonormalization of the basis $\{1,x,x^2,\ldots\}$ with respect to the inner product $\dotp{f}{g}_{\cN(0,1)}=\E_{X \sim \cN(0,1)}[f(X)g(X)]$. Every function  $f \in L^2(\R,\cN(0,1))$ can be uniquely written as $f(x) = \sum_{i \in \N} a_i h_i(x)$ and we have $\lim_{n \rightarrow n}\E_{x \sim \cN(0,1)}[(f(x)- \sum_{i =0}^n a_i h_i(x))^2] = 0$ (see, e.g., \cite{andrews_askey_roy_1999}).

\paragraph{Ornstein-Uhlenbeck Operator}
For a $\rho > 0$, we define the \emph{Gaussian noise} (or \emph{Ornstein-Uhlenbeck}) operator $U_\rho$ as the operator that maps a distribution $F$ on $\R$ to the distribution of the random variable $\rho X + \sqrt{1-\rho^2}Z$, where $X \sim F$ and $Z \sim \cN(0,1)$ independently of $X$. 
A well-known property of \emph{Ornstein–Uhlenbeck} operator is that it operates diagonally with respect to Hermite polynomials.
\begin{fact}[see, e.g.,~\cite{AoBF14}] \label{fact:eigenfunction} 
    For any Hermite polynomial $h_i$, any distribution $F$ on $\R$, and $\rho\in(0,1)$, it holds that $\E_{X \sim U_\rho F}[h_i(X)] = \rho^i \E_{X \sim F}[h_i(X)]$.
\end{fact}

\paragraph{Background on the SQ Model} 
We provide the basic definitions and facts that we use.

\begin{definition}[Search problems over distributions] \label{def:search}
Let $\cD$ be a set of distributions over $\R^d$, $\mathcal{F}$ be a set called solutions, 
and $\mathcal{Z} : \cD \to 2^\mathcal{F}$ be a map that assigns sets of solutions to distributions of $\cD$. The \emph{distributional search problem} $\mathcal{Z}$ over $\cD$ and $\mathcal{F}$ 
is to find a valid solution $f \in \mathcal{Z}(D)$ given statistical query oracle access 
to an unknown $D \in \cD$.
\end{definition}

The hardness of these problems is conveniently captured by the SQ dimension. 
For this, we first need to define the notion of correlation between distributions.

\begin{definition}[Pairwise Correlation] \label{def:pc}
The pairwise correlation of two distributions with probability density functions
$D_1, D_2 : \R^d \to \R_+$ with respect to a reference distribution with 
density $R: \R^d \to \R_+$, where the support of $R$ contains 
the supports of $D_1$ and $D_2$, is defined as
$\chi_{R}(D_1, D_2) \eqdef \int_{\R^d} D_1(x) D_2(x)/R(x)\, \dx - 1$.
When $D_1=D_2$, the pairwise correlation becomes the same as the 
$\chi^2$-divergence between $D_1$ and $R$, i.e., $\chi^2(D_1,R) \eqdef \int_{\R^d}D^2_1(x)/R(x) \dx - 1$.
\end{definition}

\begin{definition} \label{def:correlated} For $\gamma,\beta > 0$,
the set of distributions $\mathcal{D} = \{D_1,\ldots,D_m\}$ is called $(\gamma,\beta)$-correlated
relative to the distribution $R$ if  
$\abs{\chi_R(D_i,D_j)} \leq \gamma$, if $i\neq j$, and $\abs{\chi_R(D_i,D_j)} \leq \beta$ otherwise.
\end{definition}

The statistical dimension of a search problem is based on the largest set of $(\gamma,\beta)$-correlated distributions assigned to each solution.

\begin{definition}[Statistical Dimension] \label{def:sqdim}
For $\gamma,\beta >0$,
 a search problem $\mathcal{Z}$ over a set of solutions $\mathcal{F}$ 
and a class $\cD$ of distributions over $X$, we define the \emph{statistical dimension} of $\mathcal{Z}$, denoted by $\mathrm{SD}(\mathcal{Z},\gamma,\beta)$, to be the largest integer $m$ such that there exists a reference distribution $R$ over $X$ and a finite set of distributions $\mathcal{D}_R \subseteq \cD$ such that for any solution $f \in \mathcal{F}$, the set $\cD_f = \cD_R \setminus \mathcal{Z}^{-1}(f)$ is $(\gamma,\beta)$-correlated relative to $R$ and $|\cD_f|\geq m$.
\end{definition}

\begin{lemma}[Corollary 3.12 in~\cite{FeldmanGRVX17}]\label{cor:cor3-12-FGR}
Let $\cZ$ be a search problem over a set of solutions $\cF$ and a class of distributions $\cD$ over $\R^d$. For $\gamma,\beta>0$, let $s=\mathrm{SD}(\cZ,\gamma,\beta)$ be the statistical dimension of the problem. For any $\gamma'>0$, any $SQ$ algorithm for $\cZ$ requires either $s \gamma'/(\beta-\gamma)$ queries or at least one query to $\mathrm{STAT}(\sqrt{\gamma + \gamma'})$ oracle.
\end{lemma}

We continue by recalling the machinery from~\cite{DKS17-sq} that will be used for our construction. 

\begin{definition}[High-Dimensional Hidden Direction Distribution] \label{def:high-dim-distribution}
For a unit vector $v \in \R^d$ and a distribution $A$ on the real line with probability density function $A(x)$, define $P_{A,v}$ to be a distribution over $\R^d$, where $P_{A,v}$ is the product distribution whose orthogonal projection onto the direction of $v$ is $A$, 
and onto the subspace perpendicular to $v$ is the standard $(d{-1})$-dimensional normal distribution. 
That is, $P_{A,v}(x) := A(v^Tx) \phi_{\bot v}(x)$, where $\phi_{\bot v}(x) = \exp\left(-\|x - (v^Tx)v\|_2^2/2\right)/(2\pi)^{(d-1)/2}$.\end{definition}

The distributions $\{P_{A,v}\}$ defined above are shown to be nearly uncorrelated as long as the directions where $A$ is embedded are pairwise nearly orthogonal.

\begin{lemma}[Lemma 3.4 in~\cite{DKS17-sq}] \label{lem:lemma-3-4-DKS17}
Let $m \in \Z_+$. Let $A$ be a distribution over $\R$ that agrees with the first $m$ moments of $\cN(0,1)$. For any $v$, let $P_{A,v}$ denote the distribution from Definition~\ref{def:high-dim-distribution}. 
For all $v,u \in \R^d$, we have that $\chi_{\cN(0,I_d)}(P_{A,v},P_{A,u}) \leq |u^Tv|^{m+1} \chi^2(A,\cN(0,1))$.
\end{lemma}
The following result shows that there are exponentially many nearly-orthogonal unit vectors.
\begin{lemma}[see, e.g., Lemma 3.7 in~\cite{DKS17-sq}] \label{lem:orthogonal_vectors}
	For any $0<c<1/2$, there is a set $S$, of at least $2^{\Omega(d^c)}$ unit vectors in $\R^d$, such that for each pair of distinct $v,v' \in S$, it holds $|v^Tv'| \leq O(d^{c-1/2})$.
\end{lemma}

\subsection{Prior and Related Work}\label{ssec:related}

Early work in robust statistics, starting with the pioneering works of Huber and Tukey~\cite{Huber64, Tukey75}, 
pinned down the sample complexity of high-dimensional 
robust estimation with a minority of outliers. In contrast, until relatively recently, 
even the most basic computational questions in this field were poorly understood. 
Two concurrent works~\cite{DKKLMS16, LaiRV16} gave the first provably robust and 
efficiently computable estimators for robust mean and covariance estimation.
Since the dissemination of these works, there has been a flurry of
activity on algorithmic robust estimation in a variety of high-dimensional settings;
see~\cite{DK19-survey} for a recent survey on the topic.
Notably, the robust estimators developed in~\cite{DKKLMS16} are scalable in practice 
and yield a number of  applications in exploratory data analysis~\cite{DKK+17} and 
adversarial machine learning~\cite{TranLM18, DKK+19-sever}

The list-decodable learning setting studied in this paper was first considered in~\cite{CSV17} with a focus on mean estimation.
\cite{CSV17} gave a polynomial-time algorithm with near-optimal statistical guarantees for list-decodable mean estimation 
under a bounded covariance assumption on the clean. Subsequent work has led to significantly faster 
algorithms for the bounded covariance setting~\cite{DiakonikolasKK20, CherapanamjeriMY20, DKKLT20, DKKLT21} 
and polynomial-time algorithms with improved error guarantees 
under stronger distributional assumptions~\cite{DKS18-list, KSS18-sos}.
More recently, a line of work developed list-decodable learners for more challenging tasks, 
including linear regression~\cite{KKK19-list, RY19-list} and subspace recovery~\cite{RY20-subspace, BK20-subspace}.

\section{Information-Theoretic Bounds}  \label{sec:info-theoretic}
\subsection{Upper Bound on Sample Complexity} \label{sec:appendix-upper-bound}

In this section, we show that  $n=\poly(d,1/\alpha)$ samples suffice for list-decodable linear regression. 
\begin{theorem} \label{thm:sample-ub}
There is a (computationally inefficient) algorithm that uses $O(d/\alpha^3)$ samples from a $(1{-\alpha})$-corrupted version of a Gaussian linear regression model of Definition~\ref{def:glr} with $S' = \R^d$, and returns a list $\cL$ of $|\cL|=O(1/\alpha)$ many hypotheses such that with high probability at least one of them is within $\ell_2$-distance $O((\sigma/\alpha)\sqrt{\log(1/\alpha)})$ from the regression vector.
\end{theorem}
The proof strategy is similar to \cite{DKS18-list}. When $S$ is a set, we use the notation $X \sim_u S$ to denote that $X$ is distributed according to the uniform distribution on $S$.
We require the following theorem:
\begin{fact}[VC Inequality]\label{thm:vc}
    Let $\cF$ be a class of Boolean functions with finite VC dimension $\mathrm{VC}(\cF)$ and let a probability distribution $D$ over the domain of these functions. For a set $S$ of $n$ independent samples from $D$
    \begin{align*}
      \sup_{f \in \cF} \abs[\Big]{\pr_{X \sim_u S}[f(X)] - \pr_{X \sim D}[f(X)]} \lesssim \sqrt{\frac{\mathrm{VC}(\cF)}{n}} + \sqrt{\frac{\log(1/\tau)}{n} } \;,
    \end{align*}
    with probability at least $1-\tau$.
  \end{fact}

\begin{proof}[Proof of Theorem~\ref{thm:sample-ub}]
Recall the notation in Definitions~\ref{def:ld} and \ref{def:glr}. 
Let $T$ be the set of points generated by the $(1{-\alpha})$-corrupted version of $D_{\beta^*}$ for some unknown $\beta^* \in \R^d$.
 Let $S_1$ be the set of points that are sampled from $D_{\beta^*}$.
Since inliers are sampled with probability $\alpha$, we have that $|S_1| \geq \alpha |T|/2$ with high probability. For a $t \geq 0$, define $\cH_{t,\gamma}$  as follows:
\begin{align}
\cH_{t, \gamma} &:= \Bigg\{ \, \beta \in \R^d:  \; \exists T' \subset T, \; |T'| = \alpha |T|/2 ,  
\label{EqCond0}
\\
     &\qquad \qquad \pr_{(X,y)\sim_u T'}[ |y-X^T \beta| > \sigma t]\leq \alpha/20, \label{EqCond1}\\
     &\qquad \qquad \forall v \in \cS^{d-1}, \gamma' \geq \gamma: \; \pr_{(X,y)\sim_u T'}[ |y - X^T \beta - \gamma' v^TX|\leq \sigma t] \leq \
     \alpha/20 \label{EqCond2}\Bigg\} \;.
\end{align}

Recall that the distribution of inliers is  $X \sim \cN(0,I_d)$ and $y = X^T \beta^* + \eta$, where $\eta \sim \cN(0,\sigma^2)$ independent of  $X$. 
If $|S_1| \geq Cd /\alpha^2$ for a sufficiently large constant $C$, then we claim that $\beta^* \in \cH_{t,\gamma}$ with $t = \Theta(\sqrt{\log(1/\alpha)})$ and $ \gamma = 40 \sigma t / \alpha = \Theta((\sigma/ \alpha) \sqrt{\log(1/\alpha)})$.
Let $S'$ be a set of i.i.d. points sampled from $D_{\beta^*}$ with $|S'| = |T|\alpha/2$. 
We first argue that conditions~\eqref{EqCond1} and \eqref{EqCond2} hold under $(X,y) \sim D_{\beta^*}$, even after replacing $\alpha/20$ with $\alpha/40$ in conditions~\eqref{EqCond1} and \eqref{EqCond2}, with the claimed bounds on $t$ and $\gamma$, and then the required result on $(X,y) \sim_u S'$  will follow from the VC inequality.
Since $y- X^T \beta^* \sim \cN(0,\sigma^2)$ under $D_{\beta^*}$, we get that $\pr[|y-X^T\beta^*|> \sigma t] \leq \alpha/40$  because of Gaussian concentration. 
Let $G \sim \cN(0,1)$ independent of $\eta$.
For condition~\eqref{EqCond2}, 
the expression again reduces to concentration of a Gaussian distribution:
\begin{align*}
\pr_{\eta \sim \cN(0,\sigma^2), G \sim \cN(0,1)} [|\eta + \gamma' G| \leq \sigma t] = \pr_{Z \sim \cN(0, \sigma^2 + \gamma'^2)} [|Z|\ \leq \sigma t] \lesssim  \frac{\sigma t}{\gamma'},
\end{align*}
which is less than $\alpha/40$ for $\gamma' \geq \gamma = 40t \sigma /\alpha$.
The desired conclusion now follows by noting that conditions~\eqref{EqCond1} and \eqref{EqCond2} follow by uniform concentration of linear threshold functions on $(X,y)$, which have VC dimension $O(d)$ and the condition that $|S'| = \Omega(d/\alpha^2)$.

We then show that any $\gamma$-packing of the set $\cH_{t, \gamma}$ has size $O(1/\alpha)$. Having this, it follows that there exists a $\gamma$-cover of size $O(1/\alpha)$ and the output of the algorithm, $\cL$, consists of returning any such cover. 
The key claim for bounding the size of any $\gamma$-packing is that the pairwise intersections between the sets $T'$
from condition~\eqref{EqCond0} are small.

\begin{claim} Let $\beta_1,\dots, \beta_k \in \cH_{t, \gamma}$ such that $\|\beta_i - \beta_j\|_2 > \gamma$ for all $i, j \in [k]$ and $i \neq j$. Let $T'_i$ be the corresponding subsets of $T$ satisfying the condition~\eqref{EqCond0}. Then $|T'_i \cap T'_j| \leq  \alpha(|T'_i| + |T'_j|)/20$.
\end{claim}
\begin{proof}
Fix an $i \neq j$.  
Let $\beta_i - \beta_j = v \gamma'$, where $v \in \cS^{d-1}$ and $\gamma' \geq \gamma$.
Let $\cE$ be the  event $\{ (X,y): |y-X^T \beta_j| \leq \sigma t\}$ and $\cE^c$ be its complement.
As $T'_i$ and $T'_j$ are sets of size $\alpha |T|/2$, we have that
\begin{align*}
|T'_i \cap T'_j| &= \frac{|T'_i| + |T'_j|}{2} \left( \frac{|T'_i \cap T'_j \cap \cE|}{ |T'_i|} + \frac{|T'_i \cap T'_j \cap \cE^c|}{ |T'_j|} \right)\\
&\leq \frac{|T'_i| + |T'_j|}{2} \left(\frac{|T'_i \cap \cE|}{|T'_i| } + \frac{|T'_j \cap \cE^c|}{ |T'_j|} \right) = \frac{|T'_i| + |T'_j|}{2}\left( \pr_{(X,y)\sim_uT'_i}[\cE] + \pr_{(X,y)\sim_uT'_j}[\cE^c]\right).
\end{align*}
As $\beta_j \in \cH_{t, \gamma}$, we have that $\p_{(X,y)\sim_uT'_j}[\cE^c] \leq \alpha/20 $ by condition~\eqref{EqCond1}.
 We now bound the first term.
\begin{align*}
\pr_{(X,y)\sim_uT'_i}[\cE] &= \pr_{(X,y)\sim_uT'_i}[|y-X^T \beta_i - \gamma' v^TX | \leq \sigma t],
\end{align*}
which is less than $\alpha/20$ by the condition~\eqref{EqCond2}. This completes the proof of the claim.
\end{proof}
We use this to show that there cannot exist a $\gamma$-packing of size $k\geq 4/\alpha$. To see this, assume that $k=4/\alpha$, then 
\begin{align*}
    |T| \geq \sum_{i=1}^k |T'_i| - \sum_{1\leq i < j \leq k} |T'_i \cap T'_j| \geq \left(1-\frac{\alpha}{20}(k-1)\right)\sum_{i=1}^k|T'_i| \geq \frac{4}{5}k\alpha \frac{|T|}{2} > |T| \;.
\end{align*}
This yields a contradiction, completing the proof of Theorem~\ref{thm:sample-ub}.
\end{proof}

\subsection{Information-Theoretic Lower Bound on Error}
We establish the following lower bound on the error of any list-decoding algorithm for linear regression.
\begin{theorem}
\label{thm:sample-lb}
    Let $0<\alpha<1/2$, $\sigma>0$, $k>1$ such that $k = O(1/(\alpha^2\log(1/\alpha)))$, and $d \in \Z_+$ such that $d>(\log(1/\alpha^k))^C$, where $C$ is a sufficiently large constant. 
    Any list-decodable algorithm that receives a $(1{-\alpha})$-corrupted version of $D_{\beta}$ (defined in Definition~\ref{def:glr}) for some unknown $\beta \in \R^d$, and returns a list $\cL$ of size  $|\cL|= O((1/\alpha)^k)$  has error bound $\Omega\left( \frac{\sigma}{\alpha \sqrt{k\log(1/\alpha)}} \right)$ with high probability.
\end{theorem}
\begin{proof}
    Let $\rho>0$ to be decided later. We will take $\beta$ to be of the form $\rho v$ for some unit vector $v$.
    By abusing notation, let $D_v(x,y)$ be the joint distribution on $(X,y)$ from the linear model $X\sim \cN(0,I_d)$, $y = \beta^TX + \eta$, where $\eta \sim \cN(0,\sigma^2)$ independently of $X$ and $\beta = \rho v$. As  $d$ is large enough, let $S'$ be a subset of the set $S$ of nearly orthogonal unit vectors of $\R^d$ from Lemma~\ref{lem:orthogonal_vectors} with $|S'|=\lfloor 0.5(1/\alpha)^k \rfloor$ for $k>1$. Consider the set of distributions $\{ D_v \}_{v \in S'}$ and note that for every distinct pair  $u,v \in S$ we have that $\snorm{2}{\rho u - \rho v}  \geq c\rho$ for some $c>0$. We want to show that after adding $(1-\alpha)$-fraction of outliers these distributions become indistinguishable, i.e., there exists some distribution that is pointwise greater than $\alpha D_v$ for every $v \in S'$. This will lead to a lower bound on error of the form $\Omega(\rho)$.
    Let $P$ be the joint pseudo-distribution on $(X,y)$ such that $P(x,y) = \max_{v \in S} D_v(x,y)$ and denote by $\|P\|_1$ the normalizing factor $\int_{\R} \int_{\R^{d}} P(x,y) \dx \dy$. We will show that $P/\|P\|_1 \geq \alpha D_v$ pointwise. To this end, it suffices to show that $\|P\|_1 \leq 1/\alpha$. Denote $z := v^Tx$. Noting that $D_v$'s marginal on $x$ is $\cN(0,I_d)$ and the conditional $D_v(y|x)$ is $\cN(\rho z, \sigma^2)$ we can write
    \begin{align*}
        D_v(x,y) &=  \frac{1}{\sqrt{2 \pi }\sigma} \exp\left(-\frac{|y - \rho z |^2}{2\sigma^2}\right) \frac{1}{(\sqrt{2\pi })^d} \exp \left(- \frac{\|x\|^2}{2}\right) \\
                                     &= \frac{1}{(\sqrt{2 \pi })^{d+1}\sigma} \exp\left(-\frac{|y - \rho z |^2}{2\sigma^2} - \frac{\|x\|^2}{2}\right).
    \end{align*}

For some $\sigma_1$ to be defined later, take $R$ to be the reference distribution where $X \sim \cN(0,I_d)$ and $y \sim \cN(0,\sigma_1^2)$ independently. We now calculate the ratio of density of $R$ with $D_v$ at arbitrary $(x,y)$:
    
    \begin{align*}
        \frac{R(x,y)}{ D_v(x,y)} &= \frac{R(y)R(x|y)}{ D_v(y) D_v(x|y)}\\
                                        &= \frac{\frac{1}{(\sqrt{2 \pi })^{d+1} \sigma_1 } \exp\left(-0.5\|x\|^2- 0.5y^2/\sigma_1^2 \right)}{\frac{1}{(\sqrt{2 \pi })^{d+1} \sigma } \exp\left(-0.5\|x\|^2- 0.5\frac{\rho^2}{\sigma^2}\left(z - \frac{y}{\rho}\right)^2 \right)} \\
                                        &=  \frac{ \sigma }{\sigma_1} \exp\left( - \frac{y^2}{2\sigma_1^2}  + \frac{\rho^2}{2\sigma^2}\left(z - \frac{y}{\rho}\right)^2 \right) \\
                                        &\geq  \frac{ \sigma }{\sigma_1} \exp\left( - \frac{y^2}{2\sigma_1^2} \right).                                    
     \end{align*}
    As we will show later, it suffices to show that this expression is greater than $2\alpha$ with high probability under $D_v$.
     As $y \sim \cN(0,\sigma_y^2)$ under $D_v$, with probability $1 - \alpha^{k-1}$, $|y| \leq 10 \sqrt{k} \sigma_y \sqrt{\log(1/\alpha)}$.
     Setting $\sigma_1 = 10 \sqrt{k} \sigma_y \sqrt{ \log(1/\alpha)}$, we get that with the same probability, 
     \begin{align*}
        \frac{R(x,y)}{ D_v(x,y)} \geq \frac{1}{100 \sqrt{k} } \frac{\sigma}{\sigma_y \sqrt{ \log(1/\alpha)}}.
     \end{align*}
    We can now try to maximize $\rho$ (and thus $\sigma_y$) so that the expression on the right-hand side is greater than $2\alpha$. This holds as long as $\rho$ satisfies the following: 
    \begin{align*}
        \sigma_y^2 = \sigma^2 + \rho^2 \leq \frac{\sigma^2}{C' k \alpha^2 \log(1/\alpha)},
    \end{align*}
    As $k = O(1/(\alpha^2\log(1/\alpha)))$, the condition above shows that $\rho$ can be as large as $\Theta((\sigma/(\sqrt{k}\alpha))/\sqrt{\log(1/\alpha)} )$.
    Finally we show that $\|P\|_1$ is less than $1/\alpha$ as follows:
    \begin{align*}
        \|P \|_{1} &= \int_\R \int_{\R^d} P(x,y) \dx \dy \\
        &=  \int_\R \int_{\R^d} P(x,y) \mathbf{1}(|y| \leq 10\sqrt{k} \sigma_y \sqrt{\log(1/\alpha)}) \dx \dy +  \int_\R \int_{\R^d} P(x,y) \mathbf{1}(|y|>10 \sqrt{k} \sigma_y \sqrt{\log(1/\alpha)}) \dx \dy\\
        &\leq \frac{1}{2\alpha} \int_\R \int_{\R^d} R(x,y) \dx \dy + \int_\R \int_{\R^d} P(x,y) \mathbf{1}(|y|>10\sqrt{k} \sigma_y \sqrt{\log(1/\alpha)}) \dx \dy\\
        &\leq \frac{1}{2\alpha} + \sum_{v \in S'} \pr_{(X,y) \sim D_v} \left[|y|>10 \sqrt{k} \sigma_y \sqrt{\log(1/\alpha)} \right] \\
        &\leq \frac{1}{2\alpha} + |S'| \alpha^{k-1} \leq 1/\alpha,
    \end{align*}  
    where the last inequality follows by noting that $|S'|\leq 0.5 (1/\alpha)^k$. 
\end{proof}

\section{Main Result: Proof of Theorem~\ref{thm:sq-lb-inf}} \label{sec:proof}

In this section, we present the main result of this paper: 
SQ hardness of list-decodable linear regression (Definitions~\ref{def:ld} and~\ref{def:glr}). 
We consider the setting when $\beta$ has norm less than $1$, i.e.,
$\beta = \rho v$ for $v \in \cS^{d-1}$ and $\rho \in (0,1)$.\footnote{This is a standard assumption and considered by existing works~\cite{KKK19-list,RY19-list} (cf. Remark~\ref{remark:comparison}).} 
Note that the marginal distribution of the labels is $\cN(0,\sigma_y^2)$, 
where $\sigma_y^2 = \rho^2 + \sigma^2$. 
We ensure that the labels $y$ have unit variance by using $\sigma^2=1-\rho^2$.
Specifically, the choice of parameters will be such that obtaining a $\rho/2$-additive approximation of the regressor $\beta$ is possible information-theoretically with poly$(d/\alpha)$ samples (cf. Section~\ref{sec:appendix-upper-bound}), but the complexity of any SQ algorithm for the task must necessarily be at least $d^{\poly(1/\alpha)}/\sigma$. We show the following more detailed statement of Theorem~\ref{thm:sq-lb-inf}. 

\begin{theorem}[SQ Lower Bound] \label{thm:main-formal}
Let $c \in (0,1/2)$, $d \in \Z_+$ with $d = 2^{\Omega(1/(1/2-c))}$, $\alpha \in (0,1/2)$, $\rho \in (0,1)$, 
$\sigma^2=1-\rho^2$, and $m \in \Z_+$ with $m \leq c_1/\sqrt{\alpha}$ for some sufficiently small constant $c_1>0$. 
Any list-decoding algorithm that, given statistical query access to a $(1{-\alpha})$-corrupted version 
of the distribution described by the model of Definition~\ref{def:glr} with $\beta=\rho v$ for $v\in\cS^{d-1}$,  
returns a list $\cL$ of hypotheses vectors that contains a $\widehat{\beta}$ such that $\|\widehat{\beta} - \beta\|_2 \leq \rho/2$, 
does one of the following: 
\begin{itemize}
    \item it uses at least one query to  $\mathrm{STAT}\left(\Omega(d)^{-(2m+1)(1/4-c/2)} e^{O(m)}/\sigma\right)$,
    \item it makes $2^{\Omega(d^{c})} d^{-(2m+1)(1/2-c)}$ many queries, or 
    \item it returns a list $\cL$ of size $2^{\Omega(d^{c})}$.
\end{itemize}

\end{theorem}
    
In the rest of this section, we will explain the hard-to-learn construction for our SQ lower bound, i.e., a set of distributions with large statistical dimension. The proof would then follow from Lemma~\ref{cor:cor3-12-FGR}.
We begin by describing additional notation that we will use.

\noindent \textbf{Notation}: As $\beta = \rho v$ for a fixed $\rho$, we will slightly abuse notation by using $D_v(x,y)$ 
to denote the joint distribution of the inliers and we use $E_v(x,y)$ to denote the $(1{-\alpha})$-corrupted 
version of $D_v(x,y)$. To avoid using multiple subscripts, we  use $D_v(x|y)$ to denote the conditional 
distribution of $X|y$ according to the distribution $D_v$ and similarly for the other distributions. 
In addition, we use $D_v(y)$ to denote the marginal distribution of $y$ under $D_v$ and similarly for other distributions.

Following the general construction of~\cite{DKS17-sq}, we will specify a \emph{reference} joint distribution 
$R(x,y)$ where $X$ and $y$ are independent, and $X \sim \cN(0,I_d)$. 
We will find a marginal distribution $R(y)$ and a way to add the outliers 
so that the following hold for each $E_v$ (where $m=\Theta(1/\sqrt{\alpha})$):
\begin{enumerate}[label=(\Roman*), leftmargin = *]
	\item $E_v$ is indeed a valid distribution of $(X,y)$ in our corruption model 
	(i.e., can be written as a mixture $\alpha D_v(x,y) + (1{-\alpha})N_v(x,y)$ for some noise distribution $N_v$). 
	Moreover, the marginal of $E_v$ on the labels, $E_v(y)$, coincides with $R(y)$. \label{prop:E1}
	\item For every $y \in \R$, the conditional distribution $E_v(x|y)$ is of the form $P_{A_y,v}$ of Definition~\ref{def:high-dim-distribution}, 
	with $A_y$ being a distribution that matches the first $2m$ moments 
	with $\cN(0,1)$.\footnote{We use even number of moments for simplicity. The analysis would slightly differ for odd number.}  \label{prop:E2}
	\item For $A_y$ defined above,
	$\E_{y \sim R(y)}[\chi^2(A_y,\cN(0,1))]$ is bounded.    \label{prop:E3}
\end{enumerate}

We first briefly explain why a construction satisfying the above properties suffices 
to prove our main theorem (postponing a formal proof for the end of this section). 
We start by noting the following decomposition. %
\begin{lemma} \label{lem:form_of_conditional_distr}
	For $u,v \in \cS^{d-1}$, if
	$E_u$ and $E_v$ have the same marginals $R(y)$ on the labels, they satisfy
	$
		\chi_{R(x,y)}(E_v(x,y), E_u(x,y)) = \E_{y \sim R(y)} \left[ \chi_{\cN(0,I_d)}\left(E_v(x|y), E_u(x|y)\right) \right].
	$
\end{lemma}
\begin{proof}
	Let $\phi$ denote the density of $\cN(0,1)$. Using the fact that $E_v$ and $E_u$ have the same marginal $R(y)$ we have that
	\begin{align*}
		\chi_{R(x,y)}(E_v(x,y), E_u(x,y)) + 1 &= \int_{\R}  \int_{\R^d} \frac{E_v(x,y) E_u(x,y)}{\phi(x) R(y)} \dx \dy \\
		&=  \int_{\R}  \int_{\R^d} \frac{E_v(x|y) E_u(x|y) }{\phi(x)} R(y) \dx \dy \\
		&=  \int_{\R}  \left(1 + \chi_{\cN(0,I_d)}(E_v(x|y) E_u(x|y)) \right)  R(y)  \dy \\
		&= 1 + \E_{y \sim R(y)} \left[ \chi_{\cN(0,I_d)}\left(E_v(x|y), E_u(x|y)\right) \right] \;.  \qedhere
	\end{align*} 
\end{proof}
Using the decomposition in Lemma~\ref{lem:form_of_conditional_distr} for $E_u$ and $E_v$ satisfying Property~\ref{prop:E2},  Lemma~\ref{lem:lemma-3-4-DKS17} implies that 
$|\chi_{R(x,y)}(E_v(x,y), E_u(x,y))|\leq |u^Tv|^{2m+1}\E_{y \sim R(y)}[\chi^2(A_y,N(0,1))]$.
Letting $\cD = \{E_v: v \in S\}$, where $S$ is the set of nearly uncorrelated unit vectors from 
Lemma~\ref{lem:orthogonal_vectors}, we get that $\cD$ is $(\gamma, b)$-correlated relative to $R$, 
for $b = \E_{y\sim R(y)}[\chi^2(A_y,\cN(0,1))]$ and $\gamma \leq d^{-\Omega(m)}$b. 
As $|S| = 2^{\Omega(d^c)}$, $b$ is bounded by Property~\ref{prop:E3}, and the list size is much smaller than $|S|$, 
we can show that the statistical dimension of the list-decodable linear regression is large.

Thus, in the rest of the section we focus on showing that such a construction exists.
We first note that according to our linear model of Definition~\ref{def:glr}, 
the conditional distribution of $X$ given $y$ for the inliers is  Gaussian 
with unit variance in all but one direction.%
\begin{fact} \label{fact:inliers_conditional_dist}
Fix $\rho>0$, $v \in \cS^{d-1}$, and consider the regression model of Definition~\ref{def:glr} with $\beta = \rho v$. 
Then the conditional distribution $X|y$ of the inliers is $\cN(y\rho v, I_d - \rho^2 v v^T)$, 
i.e., independent standard Gaussian in all directions perpendicular to $v$ 
and $\cN(\rho y, 1-\rho^2)$ in the direction of $v$.
\end{fact}
\begin{proof}
    This is due to the following fact for the conditional distribution of the Gaussian distribution.
    \begin{fact}
        If $\left[\begin{matrix} y_1 \\ y_2 \end{matrix} \right] \sim \cN\left(\left[\begin{matrix} \mu_1 \\ \mu_2 \end{matrix} \right] , \left[ \begin{matrix} \Sigma_{11} &\Sigma_{12} \\ \Sigma_{21} &\Sigma_{22} \end{matrix} \right] \right)$, then $y_1 | y_2 \sim \cN(\bar{\mu}, \bar{\Sigma})$, with $\bar{\mu} = \mu_1 + \Sigma_{12}\Sigma_{22}^{-1}(y_2 - \mu_2)$ and $\Sigma_{11} - \Sigma_{12} \Sigma_{22}^{-1} \Sigma_{21}$.
    \end{fact}
    We apply this fact for the pair $(X,y)$ by setting $y_1=X, y_2=y, \mu_1=\mu_2=0$ and $\Sigma_{11}=I_d, \Sigma_{12}=\beta, \Sigma_{21}=\beta^T, \Sigma_{22}=\sigma^2+\snorm{2}{\beta}^2$.
\end{proof}

Since Fact~\ref{fact:inliers_conditional_dist} states that $D_v(x|y)$ 
is already of the desired form (standard normal in all directions perpendicular to $v$ and $\cN(y\rho,1-\rho^2)$ in the direction of $v$), 
the problem becomes one-dimensional. More specifically, for every $y \in \R$, we need to find a one-dimensional distribution 
$Q_y$ and appropriate values $\alpha_y \in [0,1]$ such that the mixture 
$A_y= \alpha_y \cN(y\rho,1-\rho^2) + (1{-\alpha_y} ) Q_y$
matches the first $2m$ moments with $\cN(0,1)$. Then, multiplying by $\phi_{\bot v}$ 
(which denotes the contribution of the space orthogonal to $v$ to the density of  standard Gaussian, 
as defined in Definition~\ref{def:high-dim-distribution})
yields the $d$-dimensional mixture distribution 
$\alpha_y D_v(x|y) + (1{-\alpha_y} ) Q_y(v^Tx) \phi_{\bot v}(x)$. 
We show that an appropriate selection of $\alpha_y$ 
can ensure that this is a valid distribution for our contamination model.  
\begin{lemma}\label{lem:check_properties}
Let $R$ be a distribution on pairs $(x,y)\in \R^{d+1}$ such that $\alpha_y := \alpha D_v(y)/R(y) \in [0,1]$ for all $y\in \R$. 
Suppose that for every $y \in \R$ there exists a univariate distribution $Q_y$ such that 
$A_y:=\alpha_y \cN(y\rho, 1-\rho^2)+(1{-\alpha_y})Q_y$ matches the first $2m$ moments with $\cN(0,1)$. 
If the distribution of the outliers is $N_v(x,y) =( (1{-\alpha_y})/(1{-\alpha)}) Q_y(v^Tx) \phi_{\bot v}(x) R(y)$, 
Properties~\ref{prop:E1} and~\ref{prop:E2} hold.
\end{lemma}
\begin{proof}
	First note that the noise distribution $N_v$ is indeed a valid distribution since it is non-negative everywhere because of the assumption $\alpha_y \in [0, 1]$ and it integrates to one:
	\begin{align*}
		\frac{1}{1-\alpha}\int_\R \int_{\R^d} &(1-\alpha_y)Q_y(v^Tx) \phi_{\bot v}(x) R(y) \dx \dy \\
		&= \frac{1}{1-\alpha}\left(\int_\R \int_{\R^d} R(y)Q_y(v^Tx) \phi_{\bot v}(x)  \dx \dy - \alpha \int_\R \int_{\R^d} D_v(y)Q_y(v^Tx) \phi_{\bot v}(x)  \dx \dy \right) = 1 \;.
	\end{align*}
	The joint distribution $E_v(x,y)$ can be written as 
	\begin{align*}
		E_v(x,y) &= \alpha D_v(x,y) + (1-\alpha) N_v(x,y)\\
		&= \alpha D_v(x,y) + (1-\alpha)\frac{1-\alpha_y}{1-\alpha} Q_y(v^Tx) \phi_{\bot v}(x) R(y) \\
		&= \left( \alpha_y D_v(x|y)  + (1-\alpha_y ) Q_y(v^Tx) \phi_{\bot v}(x)  \right) R(y) \;.
	\end{align*}
	This means that the marginal of $y$ under $E_v$ is $R(y)$, which establishes Property~\ref{prop:E1}, and the conditional distribution of $X|y$ under $E_v$ is $E_y(x|y)= \alpha_y D_v(x|y) + (1-\alpha_y ) Q_y(v^Tx) \phi_{\bot v}(x) $. 
	
	The moment matching part of Property~\ref{prop:E2} holds by assumption. For the other part of Property~\ref{prop:E2}, we note that $E_v(x|y)$ is standard Gaussian in directions perpendicular to $v$ because of Fact~\ref{fact:inliers_conditional_dist} and the form of the term $Q_y(v^Tx) \phi_{\bot v}(x)$ that corresponds to the outliers.
\end{proof}

We will choose the reference distribution $R(x,y)$ to have $X \sim \cN(0,I_d)$ 
and $y \sim \cN(0,1/\alpha)$ independently, which makes the corresponding value of $\alpha_y$ 
to be $\alpha_y = \alpha D_v(y)/R(y) = \sqrt{\alpha}\exp(-y^2(1 - \alpha)/2)$. 
This satisfies the condition in Lemma~\ref{lem:check_properties} that $\alpha_y \in [0,1]$.
Our choice of $R(y)$ being $\cN(0,1/\alpha)$ is informed by Properties~\ref{prop:E2} and~\ref{prop:E3}, 
and will be used later on in the proofs of Theorem~\ref{ThmHardDist_1} 
and  Lemma~\ref{lem:existence_of_A_y} (also see the last paragraph of Section~\ref{ssec:techniques} for more intuition). 
Going back to our goal, i.e., making $A_y = \alpha_y \cN(y\rho,1{-\rho^2}) + (1{-\alpha_y} ) Q_y$ 
match moments with $\cN(0,1)$, we will argue that it suffices to only look for $Q_y$ 
of the specific form $U_\rho F_y$, where $U_\rho$ is the Ornstein-Uhlenbeck operator.
This suffices because $U_\rho\delta_y = \cN(y\rho, 1-\rho^2)$ and the operator $U_\rho$ 
preserves the moments of a distribution if they match with $\cN(0,1)$ 
(see Lemma~\ref{lem:existence_of_A_y}~(i) below). 
Letting $A_y=U_\rho(\alpha_y\delta_y + (1-\alpha_y)F_y)$, 
the new goal is to show that the argument of $U_\rho$ matches moments with $\cN(0,1)$.
We show the following structural result:

\begin{theorem} \label{ThmHardDist_1}
Let $y\in \R$, $B \in \R$, $\alpha \in (0,1/2)$, and define $\alpha_y := \sqrt{\alpha}\exp(-y^2(1{ - \alpha})/2)$.
For any $m \in \Z_+$ such that $m \leq C_1/\sqrt{\alpha}$ and $B \geq C_2 \sqrt{m}$, 
with $C_1>0$ being a sufficiently small constant and $C_2$ being a sufficiently large constant, 
there exists a distribution $F_y$ that satisfies the following:
\begin{enumerate}
\item The mixture distribution $\alpha_y \delta_y + (1 - \alpha_y) F_y$ matches the first $2m$ moments with $\cN(0,1)$.
\item $F_y$ is a discrete distribution supported on at most $2m+1$ points, all of which lie in $[-B,B]$.
\end{enumerate} 
\end{theorem}
The proof of Theorem~\ref{ThmHardDist_1} is the bulk of the technical work of this paper 
and is deferred to Section~\ref{sec:the_hard_distribution}. 
As mentioned before, applying $U_\rho$ preserves the required moment-matching property. 
More crucially, it allows us to bound the $\chi^2$-divergence: 
the following result bounds $\chi^2(A_y,\cN(0,1))$ using contraction properties of $U_\rho$, 
tail bounds of Hermite polynomials, and the discreteness of $F_y$.  
	
\begin{lemma}\label{lem:existence_of_A_y}
In the setting of Theorem~\ref{ThmHardDist_1}, let $\rho>0$ and $Q_y = U_\rho F_y$. 
Then the following holds for the mixture $A_y = \alpha_y \cN(y\rho,1-\rho^2) + (1{-\alpha_y} ) Q_y$: 
(i) $A_y$ matches the first $2m$ moments with $\cN(0,1)$, 
and (ii) $\chi^2(A_y, \cN(0,1)) \leq  \alpha O(e^{y^2(\alpha-1/2)})/(1 - \rho^2) + O(e^{B^2/2})/(1 - \rho^2)$.
\end{lemma}

\begin{proof}
	The first property follows by noting that $A_y=\alpha_y \cN(y\rho,1-\rho^2) + (1{-\alpha_y}) Q_y = U_\rho (\alpha_y \delta_y + (1{-\alpha_y}) F_y)$ and using the eigenvalue property of Hermite polynomials (Fact~\ref{fact:eigenfunction}). This implies that for all $i\leq 2m$ we have that 
	\begin{align*}
		\E_{X \sim U_\rho (\alpha_y \delta_y + (1 - \alpha_y) F_y)}[h_i(X)] = \rho^i \E_{X \sim \alpha_y \delta_y + (1 - \alpha_y) F_y}[h_i(X)] = \rho^i \E_{X \sim \cN(0,1)}[h_i(X)] = \E_{X \sim \cN(0,1)}[h_i(X)] ,
	\end{align*}
	where the last equation uses that $\E_{X \sim \cN(0,1)}[h_i(X)]=0$ for $i>0$ and $\E_{X \sim \cN(0,1)}[h_0(X)]=1$. Since $\{h_i: i \in [2m]\}$ form a basis of $\cP(2m)$, the space of all polynomials of degree at most $2m$, it follows that $A_y$ continues to matches $2m$ moments with $\cN(0,1)$.

	The $\chi^2$ bound is due to the bounded support in $[-B,B]$ and the Gaussian smoothing operation and can be shown as follows. First, we need the following fact whose proof is included in Appendix~\ref{sec:appendix-extras} for completeness.
	\begin{fact} \label{fact:chi_square_formula}
		For any one-dimensional distribution $P$ that matches the first $m$ moments with $\cN(0,1)$ and has $\chi^2(P,\cN(0,1)) < \infty$ the following identity is true: $
		    \chi^2(P,\cN(0,1)) = \sum_{i=m+1}^\infty \left(\E_{X \sim P}[h_i(X)]\right)^2
		$.
	\end{fact}

	Let $M_y$ denote the distribution $\alpha_y\delta_y + (1-\alpha_y)F_y$, i.e., the mixture before applying the Ornstein-Uhlenbeck operator.  In order to apply Fact~\ref{fact:chi_square_formula} to $M_y$, we need to argue that its $\chi^2$-divergence is finite. As $F_y$ is a discrete distribution, the $U_\rho$ operator will transform it to a finite sum of Gaussians with variances strictly less than $2$.  We defer the proof of the following claim to Appendix~\ref{sec:appendix-extras}.
	
	\begin{claim} \label{claim:finite-chi-square}
	If $P = \sum_{i=1}^k \lambda_i N(\mu_i,\sigma_i^2)$ with $\mu_i \in \R$, $\sigma_i < \sqrt{2}$ and $\lambda_i \geq 0$ such that $\sum_{i=1}^k\lambda_i = 1$, we have that $\chi^2(P,\cN(0,1)) < \infty$.
	\end{claim}
	
	Using the formula of Fact~\ref{fact:chi_square_formula} and  Fact~\ref{fact:eigenfunction} for the individual terms, we get that
	\begin{align}
		\chi^2(A_y,\cN(0,1)) &= \sum_{i=2m+1}^\infty \E_{X \sim U_\rho M_y}[h_i(X)]^2  \notag
		= \sum_{i=2m+1}^\infty \rho^{2i}\E_{X \sim M_y}[h_i( X)]^2 \notag \\
		&= \sum_{i=2m+1}^\infty \rho^{2i}\left( \alpha_y h_i(y) + (1-\alpha_y)\E_{x \sim F_y}[h_i(X)] \right)^2 \notag \\
		&\leq 2 \alpha^2_y\sum_{i=2m+1}^\infty \rho^{2i} h^2_i(y) + 2 (1-\alpha_y)^2 \sum_{i=2m+1}^\infty \rho^{2i} \E_{x \sim F_y}[h_i(X)]^2  \;, \label{eq:two-terms}
	\end{align}
	where the inequality uses that $(a+b)^2 \leq 2(a^2 + b^2)$ for all $a,b \in \R$. To bound this expression from above  we will use the following tail bound for Hermite polynomials.
	\begin{lemma}[\cite{Krasikov2004}]
		\label{PropTailBoundHermite2}
		Let $h_i$ be the $i$-th normalized probabilist's Hermite polynomial. Then $\max_{x \in \R} h_k^2(x) e^{-x^2/2} = O(k^{-1/6})$.
		\end{lemma}
	
	More details on how Lemma~\ref{PropTailBoundHermite2} follows from the result of \cite{Krasikov2004} can be found in Section~\ref{sec:auxiliary_details}. For the first term of Equation~\eqref{eq:two-terms}, we have that 
	\begin{align*}
		\sum_{i=2m+1}^\infty \rho^{2i} \alpha^2_y  h^2_i(y) &\leq  \sum_{i=2m+1}^\infty \rho^{2i} \alpha e^{-y^2 + \alpha y^2} O(e^{y^2/2})
		\\
		&
		\leq \alpha O(e^{y^2(\alpha-1/2)})\sum_{i=2m+1}^\infty \rho^{2i} 
		\\
		&
		\leq \alpha O(e^{y^2(\alpha-1/2)}) \rho^{2(2m+1)}/(1-\rho^2) \;,
	\end{align*}
	where the first inequality uses Lemma~\ref{PropTailBoundHermite2} and the definition of $\alpha_y$. 
	For the second term, we use the bounded support of $F_y$ in $[-B,B]$ along with the bound of Lemma~\ref{PropTailBoundHermite2} to obtain
	\begin{align*}
		 \sum_{i=2m+1}^\infty \rho^{2i} \E_{x \sim F_y}[h_i(X)]^2 \leq \sum_{i=2m+1}^\infty \rho^{2i} O(e^{B^2/2})  
		\leq O(e^{B^2/2}) \sum_{i=2m+1}^\infty \rho^{2i} 
		\leq O(e^{B^2/2}) \frac{\rho^{2(2m+1)}}{1-\rho^2}   \;.
	\end{align*}
	This completes the proof of Lemma~\ref{lem:existence_of_A_y}.
\end{proof}
Putting everything together, we now prove our main theorem.
\begin{proof}[Proof of Theorem~\ref{thm:main-formal}]
    	We will show that the following search problem $\cZ$  has large statistical dimension:
	 $\cD$ is the set of distributions of the form $E_v(x,y) = \alpha D_v(x,y) + (1{-\alpha})N_v(x,y)$ for every $v \in \cS^{d-1}$ and noise distribution $N_v$ as in Lemma~\ref{lem:check_properties}. The reference distribution $R$ is $R=\cN(0,I_d) \times \cN(0,1/\alpha)$.
	 Let $\beta(v)=\rho v$ denote the regression vector corresponding to $E_v$.
	 The set of solutions $\cF$ is the set of all lists of size $\ell$ containing vectors of norm $\rho$ in $\R^d$ and the solution set $\cZ(E_v)$ for the distribution $E_v$ is exactly the set of lists from $\cF$ having at least one element $u$ at distance $\snorm{2}{u-\beta(v)} \leq \rho/2 $. 
	 The appropriate subset of $\cD$ that we will consider is the one corresponding to the set $S$ of nearly orthogonal vectors of Lemma~\ref{lem:orthogonal_vectors}, $\cD_{R} = \{E_v  \}_{v \in S}$.   
	
	Note that for any $u \in \cF$, there exists at most one element $E_v$ in $\cD_R$ that satisfies $\snorm{2}{u-\beta(v)} \leq \rho/2$, since if there exists another $v'$ with $\snorm{2}{u-\beta(v')} \leq \rho/2$, then by triangle inequality $\snorm{2}{\beta(v)-\beta(v')} \leq \rho$. However, this cannot happen 
	because $|v^T (v')| \leq O(d^{c-1/2})$ for all $v,v' \in S$ together with $d = 2^{\Omega(1/(1/2-c))}$ implies that $\|\beta(v) - \beta(v') \|_2 \geq \rho \sqrt{2(1-v^T (v'))} \geq \rho$. This implies that for any solution list $L$, $|\cD_{R}  \setminus \cZ^{-1}(L)| \geq |\cD_{R}| - \ell$. We choose $\ell = |\cD_{R}|/2$.
	We now calculate the pairwise correlation of the set $\cD_{R}$. Let a pair of $u,v \in \cS^{d-1}$. 
	\begin{align*}
		\chi_{R(x,y)}(E_v(x,y), E_u(x,y)) &= \E_{y \sim R(y)} \left[ \chi_{\cN(0,I_d)}\left(E_v(x|y), E_u(x|y)\right) \right]\\
		&\leq |u^T v|^{2m+1} \E_{y \sim R(y)} \left[\chi^{2}(A_y,\cN(0,1))\right] \\
		&= |u^T v|^{2m+1} \left(O(e^{B^2/2})/(1-\rho^2) +  \int_\R   \alpha O(e^{y^2(\alpha-1/2)})\sqrt{\alpha} e^{-y^2\alpha}\dy  \right) \\
		&\leq |u^T v|^{2m+1} O(e^{B^2/2})/(1-\rho^2) \\
		&\leq \Omega(d)^{-(2m+1)(1/2-c)} O(e^{B^2/2})/(1-\rho^2) \;,
	\end{align*}
	where the first line is due to Lemma~\ref{lem:form_of_conditional_distr}, the second line is from Lemma~\ref{lem:lemma-3-4-DKS17} along with the observation that $E_v(x|y)$ is of the form $P_{A_y,v}$,  the third line comes from the second part of Lemma~\ref{lem:existence_of_A_y}, and the last one uses Lemma~\ref{lem:orthogonal_vectors}.
	Thus, by recalling that we can choose $B=C_2\sqrt{m}$ for a sufficiently large constant $C_2$, the set $\cD_{R}$ is $(\gamma,b)$-correlated with respect to $R$, where $\gamma:= \Omega(d)^{-(2m+1)(1/2-c)} e^{O(m)}/(1-\rho^2) $  and $b := e^{O(m)}/(1-\rho^2) $.
	The proof is concluded by applying Lemma~\ref{cor:cor3-12-FGR} with $\gamma' = \gamma$.
\end{proof}

We conclude this section with a note on the model and existing algorithmic results (extending the relevant discussion of Section~\ref{ssec:background}).
\begin{remark}[Comparison of SQ Lower Bound to Existing Upper Bounds] \label{remark:comparison}
	We remark that the model used in Theorem~\ref{thm:sq-lb-inf} (i.e., having a regressor with norm at most one and  additive noise with small variance) is considered in both recent works~\cite{KKK19-list,RY19-list} that provided list-decoding algorithms for the problem. In particular, these works give the following upper bounds:
	\begin{itemize}
		\item \cite{KKK19-list} considers the model  where $\snorm{2}{\beta} \leq 1$ and gives an algorithm that for every $\eps>0$, runs in time $(d/(\alpha \eps))^{O\left(\frac{1}{\alpha^8 \eps^8}\right)}$ and outputs a list of size $O(1/\alpha)$ containing a $\widehat{\beta}$ such that $\|{\widehat{\beta}-\beta}\|_2 \leq O(\sigma/\alpha) + \eps$. Note that this guarantee is better than the trivial upper bound of $1$ only if $\sigma = O(\alpha)$. To achieve error $1/4$, this algorithm runs in time $(d/\alpha)^{O\left(\frac{1}{\alpha^8}\right)}$. On the other hand, our lower bound for the complexity of any SQ algorithm becomes $\alpha d^{\Omega\left(1/\sqrt{\alpha}\right)}$.
		\item \cite{RY19-list} does not impose any constraint on $\snorm{2}{\beta}$ and gives an algorithm that runs in 
		time $(\|\beta\|_2/\sigma)^{\log(1/\alpha)}d^{O(1/\alpha^4)}$ and outputs a list of size $O((\|\beta\|_2/\sigma)^{\log(1/\alpha)})$ including a $\widehat{\beta}$ with the guarantee that $\|{\widehat{\beta}-\beta}\|_2 \leq O(\sigma/\alpha^{3/2})$. For the special  case where $\snorm{2}{\beta} \leq 1$ (and $\sigma = O( \alpha^{3/2})$ in order for the error guarantee to be meaningful), this algorithm can achieve error $1/4$ in time $(1/\alpha^{3/2})^{\log(1/\alpha)} d^{O(1/\alpha^4)}$. In comparison, our lower bound becomes $\alpha^{3/2} d^{\Omega\left(1/\sqrt{\alpha}\right)}$.
	\end{itemize}
\end{remark}

\section{Duality for Moment Matching: Proof of Theorem~\ref{ThmHardDist_1}} %
\label{sec:the_hard_distribution}
We now prove the existence of a bounded distribution $F_y$ 
such that the mixture $\alpha_y \delta_y + (1{-\alpha_y})F_y$ 
matches the first $2m$ moments with $\cN(0,1)$. 
The proof follows a non-constructive argument based on the duality between 
the space of moments and the space of non-negative polynomials.

Let $B>0$ and $m \in \Z_+$. Let $\cP(m)$ denote the class of all polynomials $p : \R \to \R$ with degree at most $m$.
Let $\cP^{\geq0}(2m,B)$ be the class of polynomials that can be represented in either the form 
$p(t) = (\sum_{i=0}^m a_i t^i)^2$ or the form $p(t) = (B^2-t^2)(\sum_{i=0}^{m-1} b_i t^i)^2$.
The intuition for $\cP^{\geq0}(2m,B)$ is that every polynomial of degree at most $2m$ 
that is non-negative in $[-B,B]$ can be written as a finite sum of polynomials from $\cP^{\geq0}(2m,B)$. 
By slightly abusing notation, for a polynomial $p(t) = \sum_{i=0}^m p_i t^i$, 
we also use $p$ to denote the vector in $\R^{m+1}$ consisting of the coefficients $(p_0,\ldots,p_m)$.
The following classical result characterizes when a vector is realizable as the moment sequence 
of a distribution with support in $[-B,B]$ (for simplicity, we focus on matching an even number of moments in the rest of this section).
\begin{theorem}[Theorem 16.1 of~\cite{karlin1953geometry}] \label{ThmDualityBdd}
Let $B>0$, $k \in \Z_+$, and $x=(x_0, x_1,\ldots, x_{2k})\in \R^{2k+1}$ with $x_0 = 1$. 
There exists a distribution with support in $[-B,B]$ having as its first $2k$ moments the sequence $(x_1,\dots,x_{2k})$ 
if and only if for all $p \in \cP^{\geq 0}(2k,B)$ it holds that $\sum_{i=0}^{2k} x_ip_i \geq 0$.
\end{theorem}
As we require the distribution to be discrete, we prove the following result using Theorem~\ref{ThmDualityBdd}:

\begin{proposition} \label{cor:duality}
Fix $y \in \R$, $\alpha_y \in (0,1)$, $B>0$, and $m \in \Z_+$. There exists a discrete distribution $F_y$ 
supported on at most $2m+1$ points in $[-B,B]$ such that 
$ \alpha_y	\delta_y + (1- \alpha_y)F_y$ matches the first $2m$ moments with $\cN(0,1)$
	if and only if $\E_{X\sim \cN(0,1)}[p(X)] \geq \alpha_yp(y)$ for all 
	$p \in \cP^{\geq 0}(2m,B)$. 
\end{proposition}

The proof of Proposition~\ref{cor:duality} is deferred to Section~\ref{app:proof_of_prop_duality}.
To prove Theorem~\ref{ThmHardDist_1}, we need to establish the condition of Proposition~\ref{cor:duality}. To this end, we first need the following two technical lemmas, whose proofs are given in Sections~\ref{sec:proof_of_LemGaussPositive} and~\ref{sec:auxiliary_details}. 
\begin{lemma}
	\label{LemGaussPositive}
	Let $m \in \Z_+$. If $B \geq C \sqrt{m}$ for some sufficiently large constant $C>0$, then for every  $q \in \cP(m)$, it holds that $B^2 \E _{X \sim \cN(0,1)}[q^2(X)] \geq 2 \E_{X \sim \cN(0,1)} [X^2 q^2(X)]$.
\end{lemma}

\begin{lemma}
	\label{LemSupRatio}
	Let $y \in \R$, $\alpha \in (0,1/2)$, $m \in \Z_+$, and $\alpha_y = \sqrt{\alpha}\exp(-y^2(1 - \alpha)/2)$. Suppose $m \leq C/\sqrt{\alpha}$ for some sufficiently small constant $C>0$. Then for all $r \in \cP(m), r \not \equiv 0$: $r^2(y)/(\E_{X\sim \cN(0,1)}[r^2(X)]) \leq 1/( 2\alpha_y)$.
\end{lemma}
\begin{proof}[Proof of Theorem~\ref{ThmHardDist_1}]
	By Proposition~\ref{cor:duality}, it remains to show that if $B \geq C_2 \sqrt{m}$, then the condition $\E_{X \sim \cN(0,1)}[p(X)] \geq  \alpha_y	 p(y)$ holds for all $p \in \cP^{\geq0}(2m,B)$. 
	Thus, it suffices to ensure that the following two inequalities hold for $X \sim \cN(0,1)$:
\begin{align}
\sup_{r\in \cP(m), r\not\equiv 0} \frac{r^2(y)}{\E [r^2(X)]} \leq \frac{1}{  \alpha_y	} \,\,\,\,\,\,\, \text{and } \,\, \sup_{q\in \cP(m-1), q \not\equiv 0} \frac{(B^2 - y^2)q^2(y)}{\E[(B^2 - X^2)q^2(X) ]} \leq \frac{1}{ \alpha_y	},
\label{EqSupRation3}
\end{align}
where we use Lemma~\ref{LemGaussPositive} to show that $\E[(B^2-X^2)q^2(X)] > 0$ for all non-zero polynomials $q \in \cP(m-1)$.
The first expression can be bounded using Lemma~\ref{LemSupRatio} when $m \leq C_1/\sqrt{\alpha}$. We  now focus on the second expression. By Lemma~\ref{LemGaussPositive}, $\E_{X \sim \cN(0,1)}[(B^2-X^2)q^2(X)] \geq 0.5\E_{X \sim \cN(0,1)}[B^2q^2(X)]$. Therefore, we have that
\begin{align*}
\sup_{q\in \cP(m-1), q\not\equiv 0} \frac{(B^2 - y^2)q^2(y)}{\E_{X \sim \cN(0,1)}[(B^2 - X^2)q^2(X) ]} \leq \sup_{q\in \cP(m-1), q\not\equiv 0} \frac{B^2q^2(y)}{\E_{X \sim \cN(0,1)}[(B^2 - X^2)q^2(X) ]} \\
\leq \sup_{q\in \cP(m-1), q\not\equiv 0} \frac{ B^2q^2(y)}{\E_{X \sim \cN(0,1)}[0.5B^2q^2(X) ]} 
= 2 \sup_{q\in \cP(m-1), q\not\equiv 0} \frac{q^2(y)}{\E_{X \sim \cN(0,1)}[q^2(X) ]}\;,
\end{align*}
where the first inequality uses that the denominator is positive and $y^2q^2(y)\geq 0$ and the second inequality uses that $\E_{X \sim \cN(0,1)}[(B^2-X^2)q^2(X)] \geq 0.5\E_{X \sim \cN(0,1)}[B^2q^2(X)]$. The expression above is of the same form as the first expression in Equation~\eqref{EqSupRation3}, and thus is also bounded above  by $1/\alpha_y$  when $m \leq C_1/\sqrt{\alpha}$ using Lemma~\ref{LemSupRatio}. This completes the proof of Theorem~\ref{ThmHardDist_1}. \qedhere
\end{proof}

\subsection{Proof of Proposition~\ref{cor:duality}} \label{app:proof_of_prop_duality}

We require the following result stating that for every distribution $Q$ with bounded support, there exists a discrete distribution $P$ with bounded support that matches the low-degree moments of $Q$.
\begin{lemma}
\label{lem:LPDualityDiscrete}
Let $B>0$, $k \in \Z_+$, and $Q$ be any distribution with support in $[-B,B]$. Then there exists a discrete distribution $P$ with the following properties: (i) the support of $P$ is contained in $[-B,B]$, (ii)  the first $k$ moments of $P$ agree with the first $k$ moments of $Q$, and  (iii) $P$ is supported on at most $k+1$
points.
\end{lemma}
\begin{proof}
Let $\cQ$ be the set of distributions on $\R$ that are supported in $[-B,B]$ and let $\cQ' \subset \cQ$ be the set of Dirac delta distributions supported in $[-B,B]$, i.e., $\cQ' = \{\delta_y: y \in [-B,B]\}$.	
	 Let $\mathcal{C} \subset \R^{k}$ and $\mathcal{C}' \subset \R^{k}$  be the set of all vectors $(x_1,\ldots, x_{k})$ whose coordinates $x_1,\ldots, x_{k}$ are the moments of a distribution in $\cQ$ and $\cQ'$ respectively, i.e.,	
	 \begin{align*}	
	     \cC &:=\{x \in \R^k: \exists Q \in \cQ: \forall i \in [k], x_i = \E_{X \sim Q}[X^i]\},\\	
	     \cC' &:=\{x \in \R^k: \exists Q' \in \cQ': \forall i \in [k], x_i = \E_{X \sim Q'}[X^i]\}.	
	 \end{align*}	
	Note that there is a bijection between $\cC'$ and $\cQ'$.	
	We now recall the following classical result stating convexity properties of $\cC$ and its relation with $\cC'$. We say a set $M$ is a convex hull of a set $M'$ if every $x \in M$ can be written as 	
	$x = \sum_{i=1}^j\lambda_i y_i$, where  $j \in \Z_+ $, $\sum_{i=1}^j \lambda_i = 1$, and for all $i \in [j]$: $\lambda_i \geq 0$, $y_i \in M'$.	
	\begin{lemma}[Theorem 7.2 and 7.3 of~\cite{karlin1953geometry}]
	\label{thm:convexHullMoments}	
	$\cC$ is convex, closed, and bounded. Moreover, $\cC$ is a convex hull of $\cC'$.	
	\end{lemma}	
	    Let $x^*= (x_1^*,\dots,x^*_k)$ be the first $k$ moments of $Q$. Since $x^* \in \cC$, Caratheodory theorem and Lemma~\ref{thm:convexHullMoments} implies that $x^*$ can be written as a convex combination of at most $k+1$  elements of $\cC'$.
	    This implies that there is a distribution, which is a convex combination of at most $k+1$ Dirac delta distributions in $\cQ'$, that matches the first $k$ moments with $x^*$. This completes the proof.
	\end{proof}

We can now prove the main result of this section.
\begin{proof}[Proof of Proposition~\ref{cor:duality}]
Let $X \sim \cN(0,1)$. We note that $F_y$ should have the moment sequence $x=(x_1,\dots,x_{2m})$ where $x_i = (\E_{X \sim \cN(0,1)}[X^i] - \alpha_yy^i)/(1 - \alpha_y)$ for $i\in[2m]$.
Theorem~\ref{ThmDualityBdd} implies that this happens if and only if for all $p=(p_0,\dots,p_{2m}) \in \cP^{\geq 0}(2m,B)$, we have that $\sum_{i=0}^{2m}x_ip_i \geq 0$.
The desired expression follows by noting that $    \sum_{i=0}^{2m}x_ip_i =  (\sum_{i=0}^{2m}  p_i \E_{X \sim \cN(0,1)}[X^i] - \alpha_y p_i y^i)/(1 - \alpha_y) = (\E_{X \sim \cN(0,1)}[p(X)] - \alpha_y p(y))/(1 - \alpha_y)$.
The result that $F_y$ is discrete follows from Lemma~\ref{lem:LPDualityDiscrete}. 
\end{proof}

\subsection{Proof of Lemma~\ref{LemGaussPositive}} \label{sec:proof_of_LemGaussPositive}
The proof of Lemma~\ref{LemGaussPositive} is a relatively straightforward application of Hölder's inequality and the Gaussian Hypercontractivity Theorem (stated below). 
For $p\in (0,\infty)$, we define the $L^p$-norm of a random variable $X$ to be $\snorm{L^p}{X} := (\E[|X|^p])^{1/p}$.
\begin{fact}[Gaussian Hypercontractivity \cite{Bog:98,nelson1973free}] \label{lem:hypercontractivity}
Let $X \sim \cN(0,1)$. If $p \in \cP(d)$ and $t\geq2$, then $\snorm{L^t}{p(X)} \leq (t-1)^{d/2}\snorm{L^2}{p(X)}$.
\end{fact}

\begin{proof}[Proof of Lemma~\ref{LemGaussPositive}]
    Let $X \sim \cN(0,1)$.
	We can assume that $q$ is a non-zero polynomial. Then it suffices to bound $B$ from above by $\sqrt{2}$ times the following expression:
	\begin{align*}
	\sup_{q \in \cP(m),  q\not\equiv  0} \sqrt{\frac{\E[X^2 q^2(X)]}{\E [q^2(X)]}} 
	&\leq \sup_{q \in \cP(m),  q\not\equiv 0} \sqrt{\frac{\left(\E [(X^{2})^{m+1}]\right)^{1/(m+1)}\left(\E [(q^2(X))^{\frac{m+1}{m}}] \right)^{\frac{m}{m+1}}}{ \E [q^2(X)]}} \\
	&= \sup_{q \in \cP(m), q\not\equiv 0} \frac{\|X\|_{L^{2m+2}}  \|q(X)\|_{L^{\frac{2m+2}{m}}}}{\|q(X)\|_{L^2}},
	\end{align*}
	where the first step uses Hölder's inequality.
	Using  standard concentration bounds for the standard Gaussian (or Fact~\ref{lem:hypercontractivity} with $p(x) = x$), we get that $\snorm{L^{2m+2}}{X}= O(\sqrt{m})$. 
	Gaussian Hypercontractivity (Fact~\ref{lem:hypercontractivity}) implies that for any polynomial of degree at most $m$ and $r > 2$, $\snorm{L^r}{q(X)} \leq (r-1)^{m/2} \snorm{L^2}{q(X)}$. For $r = (2m+2)/m$, we get that
	\begin{align*}
	\frac{\|q(X)\|_{L^\frac{2m+2}{m}}}{ \|q(X)\|_{L^2}  } \leq \left(\frac{2m+2}{m} -1 \right)^{\frac{m}{2}} = \left( 1 + \frac{2}{m}\right)^{\frac{m}{2}} \leq \exp(1).
	\end{align*}
	 Therefore, $B \geq C \sqrt{m}$ suffices for a sufficiently large constant $C$. 
\end{proof}

\subsection{Proof of Lemma~\ref{LemSupRatio}} %
\label{sec:auxiliary_details}

We first recall the result on the tails of Hermite polynomials.

\begin{lemma}[\cite{Krasikov2004}]
\label{PropTailBoundHermite}
Let $h_k$ be the $k$-th normalized probabilist's Hermite polynomial. Then $\max_{x \in \R} h_k^2(x) e^{-x^2/2} = O(k^{-1/6})$.
\end{lemma}
For completeness, we give an explicit calculation that translates the result of \cite{Krasikov2004} in our setting. 
\begin{proof}[Proof of Lemma~\ref{PropTailBoundHermite}]
We will split the analysis in two cases. First suppose the case when $k < 6$. As $h_k(\cdot)$ is a constant degree polynomial, we get that $\max_{x \in \R}h_k^2(x)\exp(-x^2/2)$ is a constant. For the rest of the proof, we will assume that $k \geq 6$.

For brevity, we will only consider the case where $k$ is even. The case  where $k$ is odd is similar.
Let $H_k(\cdot)$ be the physicist's Hermite polynomial.
Recall that we can relate $h_k(\cdot)$ with $H_k(\cdot)$ with the following change of variable: $H_k(x) = \sqrt{2^k k!}h_k(\sqrt{2} x)$.

\cite[Theorem 1]{Krasikov2004} implies the following:
\begin{align} \label{EqKrasikovEven}
\max_{x \in \R} \left((H_{k}(x))^2 e^{-x^2}\right) = O \left(\sqrt{k} k^{-1/6} \binom{k}{0.5k} k!\right) \;.
\end{align}
From Equation~\eqref{EqKrasikovEven} we obtain: 
\begin{align*}
\max_{x \in \R} 2^{k} k! h_{k}^2(\sqrt{2} x) e^{-x^2} = \max_{x \in \R} 2^{k}k! h_{k}^2(x) e^{-x^2/2} =O \left(\sqrt{k} k^{-1/6} \binom{k}{0.5k} k!\right)  \;\;.
\end{align*}
This implies the following:
\begin{align*}
\max_{x \in \R} h_{k}^2(x) e^{-x^2/2} = O \left( k^{-1/6} \sqrt{k} \binom{k}{0.5k} 2^{-k}\right)= O( k^{-1/6}),
\end{align*}
where we use that $\binom{k}{0.5k}2^{-k}/\sqrt{k} = O(1)$.
\end{proof}

\begin{proof}[Proof of Lemma~\ref{LemSupRatio}]
Let $h_i$ be the $i$-th normalized probabilist's Hermite polynomial. Since $r$ is a polynomial of degree at most $m$ and $\{h_i, i \in [m]\}$ form a basis for $\cP(m)$, we can represent $r(x) = \sum_{i=1}^m a_ih_i(x)$ for some $a_i \in \R$.
Using orthonormality of $h_i$ under the Gaussian measure, we get that $\E_{X\sim \cN(0,1)}[r^2(X)] = \sum_{i=1}^m a_i^2$. Since $r$ is a non-zero polynomial, we have that $\sum_{i=1}^m a_i^2 > 0$.
We thus have that
\begin{align*}
\sup_{r \in \cP(m), r\not\equiv 0} \frac{r^2(y)}{\E_{X\sim \cN(0,1)} [r^2(X)]} &= \sup_{a_1,\dots,a_m \in \R, \sum_{i=1}^m a_i^2 > 0 }\frac{\sum_{i=1}^m\sum_{j=1}^m a_ia_jh_i(y)h_j(y)}{\sum_{i=1}^m a_i^2}\\
&= \sup_{a_1,\ldots,a_m \in \R, \sum_{i=1}^m a_i^2 > 0 } \frac{\sqrt{\sum_{i=1}^m\sum_{j=1}^m a_i^2a_j^2} \sqrt{\sum_{i=1}^m\sum_{j=1}^m h_i^2(y)h_j^2(y)}}{\sum_{i=1}^m a_i^2} \\
&=  \sum_{i=1}^mh_i^2(y). 
\end{align*}
Therefore, we need to show that, for all $y\in\R$, $\sum_{i=1}^m \alpha_y h_i^2(y) \leq 1/2$ whenever $m \leq C/\sqrt{\alpha}$ for a sufficiently small constant $C>0$.
We will now split the analysis in two cases:

\paragraph{Case 1: $|y| \leq 1/\sqrt{\alpha}$.}
Using Lemma~\ref{PropTailBoundHermite} and the assumption that $|y|^2 \alpha \leq 1$, we can bound the desired expression as follows:
\begin{align*}
\max_{|y| \leq 1/\sqrt{\alpha}} \alpha_y	 h_i^2(y) &= \max_{|y| \leq 1/\sqrt{\alpha}} \sqrt{\alpha} \exp(y^2 \alpha/2 )\exp(-y^2/2)  h_i^2(y) \\
&\leq \sqrt{\alpha e} \sup_{y \in \R} \exp(-y^2/2)  h_i^2(y) \\
&= O(\sqrt{\alpha} i^{-1/6}).
\end{align*}
Therefore, we get the following bound on $\sum_{i}h_i^2(y)$. 
\begin{align*}
\sum_{i = 1}^m  \alpha_y h_i^2(y) = O\left(\sqrt{\alpha} \sum_{i=1}^m i^{-1/6} \right) = O (\sqrt{\alpha}m^{5/6})  \;.
\end{align*}
The last expression is less than $1/2$ when $m = O(1/\alpha^{3/5})$.

\paragraph{Case 2: $|y| \geq 1/\sqrt{\alpha}$.}

We will use rather crude bounds here.
We have the following explicit expression of $h_i(\cdot)$ (see, for example, \cite{andrews_askey_roy_1999,Szego:39}):
\begin{align*}
|h_i(x)| &= \left|\frac{He_i(x)}{\sqrt{i!}}\right| = \left|\sqrt{i!} \sum_{j=0}^{\lfloor i/2 \rfloor } \frac{(-1)^j}{j!(i - 2j)!} \frac{x^{i-2j}}{2^j}\right| = \left|\sqrt{i!}x^i \sum_{j=0}^{\lfloor i/2 \rfloor } \frac{(-1)^j}{(2j)!(i - 2j)!} x^{-2j} \frac{(2j)!}{j!2^j}\right|\\
		&\leq \sqrt{i!} |x|^i \sum_{k=0}^{i } \frac{i!}{k!(i - k)!} |x|^{-k} \leq (i|x|)^i (1 + |x|^{-1})^i = i^i(1 + |x|)^i.
\end{align*}
Therefore, we get the following relation for all $|y| > 1$, $\alpha< 0.5$, and $i \in \N$:
\begin{align*}
 \alpha_y	h_i^2(y) &=   \sqrt{\alpha}\exp(-y^2(1 - \alpha)/2)h_i^2(y)\\
				 &\leq  \sqrt{\alpha}\exp(-y^2/4)(2i)^i|y|^i\\
				 &= \sqrt{\alpha}\exp( -y^2/4 + i \log(2i|y|)).
\end{align*}
The expression above is at most $C'\sqrt{\alpha}$ for a constant $C'>0$ for all $|y| \geq c'\sqrt{i \log i}$ for a constant $c'>0$. The latter condition holds whenever $1/\sqrt{\alpha} \geq c'\sqrt{i \log i}$. It suffices that $i = O (1/\alpha^{0.9})$.
Overall, we get the following bound when $m =O(1/\alpha^{0.9})$:
\begin{align*}
\sup_{|y| > 1/\sqrt{\alpha}}\sum_{i=1}^m \alpha_yh_i^2(y) = O (  \sqrt{\alpha} m).
\end{align*}
The last expression is less than $1/2$ when $m \leq C/\sqrt{\alpha}$ for some constant $C > 0$. This completes the proof of Lemma~\ref{LemSupRatio}.
\end{proof}

\section{Hypothesis Testing Version of List-Decodable Linear Regression}\label{sec:appendix-reduction-statement}

\paragraph{Organization} 
We introduce Problem~\ref{prob:hard-hypothesis-testing}, which is a hypothesis testing problem
related to the search problem we discussed in Section~\ref{sec:proof}.
We first show the SQ-hardness of Problem~\ref{prob:hard-hypothesis-testing} in Theorem~\ref{thm:sq-hardness-hypothesis-testing}. 
In Section~\ref{sec:reduction_from_regression}, we give an efficient reduction from Problem~\ref{prob:hard-hypothesis-testing} 
to list-decodable linear regression, showing that Problem~\ref{prob:hard-hypothesis-testing} is indeed not harder 
than list-decodable linear regression. In Section~\ref{sec:hardness_again_low_degree_poly}, 
we also show the hardness of Problem~\ref{prob:hard-hypothesis-testing} against low-degree polynomial tests.

We begin by formally defining a hypothesis problem. 
\begin{definition}[Hypothesis testing] \label{def:hypothesis-testing-general}
    Let a distribution $D_0$ and a set $\mathcal{S} = \{ D_u \}_{u \in S}$ of distributions on $\R^d$. Let $\mu$ be a prior distribution on the indices $S$ of that family. We are given access (via i.i.d.\ samples or oracle) to an \emph{underlying} distribution where one of the two is true:
    \begin{itemize}
        \item $H_0$: The underlying distribution is $D_0$.
        \item $H_1$: First $u$ is drawn from $\mu$ and then the underlying distribution is set to be  $D_u$.
    \end{itemize}
    We say that a (randomized) algorithm solves the hypothesis testing problem if it succeeds with non-trivial probability (i.e.,  greater than $0.9$). 
\end{definition}

We now introduce the following hypothesis testing variant of the $(1{-\alpha})$-contaminated linear regression problem:
    \begin{problem} \label{prob:hard-hypothesis-testing}
        Let $\alpha \in (0,1/2)$, $\rho \in (0,1)$. Let $S$ be the set of $d$-dimensional nearly orthogonal vectors from Lemma~\ref{lem:orthogonal_vectors}. We are given access (via i.i.d.\ samples or oracle) to an \emph{underlying} distribution where one of the two is true:
        \begin{itemize}
            \item $H_0$: The underlying distribution is  $R = \cN(0,I_d) \times \cN(0,1/\alpha)$.
            \item $H_1$: First, a vector $v$ is chosen uniformly at random from $S$. The underlying distribution is set to be $E_v$, i.e., the ($1-\alpha$)-additively corrupted linear model of Definition~\ref{def:glr} with $\beta = \rho v$, $\sigma^2 = 1-\rho^2$, and a fixed noise distribution $N_v$ as specified in Lemma~\ref{lem:check_properties}.
        \end{itemize}
    \end{problem}
Using the reduction outlined in Lemma~\ref{lem:reduction}, it follows that $O(d/\alpha^3)$ samples suffice to solve Problem~\ref{prob:hard-hypothesis-testing} when $\sigma \leq O( \alpha/\sqrt{\log(1/\alpha)})$.
On the other hand, the following result shows an SQ lower bound of $d^{\poly{(1/\alpha)}}$. 

\begin{theorem}[SQ Hardness of Problem~\ref{prob:hard-hypothesis-testing}] \label{thm:sq-hardness-hypothesis-testing}
    Let $0<c<1/2$, $m \in \Z_+$ with $m \leq c_1/\sqrt{\alpha}$ for some sufficiently small constant $c_1>0$ and $d = m^{\Omega(1/c)}$. Every SQ algorithm that solves Problem~\ref{prob:hard-hypothesis-testing} either performs $2^{\Omega(d^{c/4})}$ queries or performs at least one query to $\mathrm{STAT}\left(\Omega(d)^{-(2m+1)(1/4-c/2)} e^{O(m)}/\sigma\right)$.
\end{theorem}
We note that the lower bound on the (appropriate) statistical dimension implies SQ hardness of the (corresponding) hypothesis testing problem. As the Problem~\ref{prob:hard-hypothesis-testing} differs slightly from the kind of hypothesis testing problems considered in \cite{FeldmanGRVX17}, we provide the proof of Theorem~\ref{thm:sq-hardness-hypothesis-testing} in Section~\ref{sec:hardness_hyp_testing_sq}, where we introduce the relevant statistical dimension  from \cite{brennan2020statistical} (Definition~\ref{def:SDA} in this paper). %

\subsection{Hardness of Problem~\ref{prob:hard-hypothesis-testing} in the SQ Model}
\label{sec:hardness_hyp_testing_sq}

We  need the following variant of the statistical dimension from \cite{brennan2020statistical}, which is closely related to the hypothesis testing problems considered in this section. Since this is a slightly different definition from the statistical dimension ($\mathrm{SD}$) used so far, we will assign the distinct notation ($\mathrm{SDA}$) for it.

\paragraph{Notation} For $f:\R \to \R$, $g:\R \to \R$ and a distribution $D$, we define the inner product $\langle f,g \rangle_D = \E_{X \sim D}[f(X)g(X)]$ and the norm $\snorm{D}{f} = \sqrt{\langle f,f \rangle_D }$.

\begin{definition}[Statistical Dimension] \label{def:SDA}
    For the hypothesis testing problem of Definition~\ref{def:hypothesis-testing-general}, we define the \emph{statistical dimension} $\mathrm{SDA}(\mathcal{S},\mu, n)$ as follows:
    \begin{align*}
        \mathrm{SDA}(\mathcal{S},\mu, n) = \max \left\lbrace q \in \mathbb{N}  :  \E_{u,v \sim \mu}[| \langle  \bar{D}_u,\bar{D}_v  \rangle_{D_0} - 1 | \; | \; \cE] \leq \frac{1}{n} \; \text{for all events $\cE$ s.t. } \pr_{u,v \sim \mu}[\cE] \geq \frac{1}{q^2} \right\rbrace \;.
    \end{align*}
We will omit writing $\mu$ when it is clear from the context.
\end{definition}

\begin{theorem}[Theorem A.5 of~\cite{brennan2020statistical}]
\label{thm:Sq_hardness_from_SDA}
    Let $\cS = \{D_u\}_{u \in S}$ vs. $D_0$ be a hypothesis testing problem with prior $\mu$ on $\cS$. If $\mathrm{SDA}(\cS,\mu,3/t) > q$, then every SQ algorithm that solves the hypothesis testing problem
    either makes at least $q$ queries, or makes at least one query to $\mathrm{STAT}(\sqrt{t})$.
\end{theorem}

In order to prove Theorem~\ref{prob:hard-hypothesis-testing}, we will prove a lower bound on the $\mathrm{SDA}$ of Problem~\ref{prob:hard-hypothesis-testing}.  
As we will show later, Problem~\ref{prob:hard-hypothesis-testing} is a special case of the following hypothesis testing problem:
\begin{problem}[Non-Gaussian Component Hypothesis Testing] \label{prob:non-gaussian-component-testing}
     Let $R$ be the joint distribution $R$ over the pair $(X,y) \in \R^{d+1}$ where $X\sim \mathcal{N}(0,I_d)$ and $y \sim R(y)$ independently of $X$. Let $E_v$ be the joint distribution over pairs $(X,y) \in \R^{d+1}$ where the marginal on $y$ is again $R(y)$ but the conditional distribution $E_v(x|y)$ is of the form $P_{A_y,v}$ (with $P_{A_y,v}$ as in Definition~\ref{def:high-dim-distribution}). Define $\cS = \{E_v \}_{v \in S}$ for $S$ being the set of $d$-dimensional nearly orthogonal vectors from Lemma~\ref{lem:orthogonal_vectors} and let the hypothesis testing problem  be distinguishing between $R$ vs. $\mathcal{S}$ with prior $\mu$ being the uniform distribution on $S$.
\end{problem}

The following lemma translates the $(\gamma,\beta)$-correlation of $\cS$ to a lower bound for the statistical dimension of the hypothesis testing problem.
The proof is very similar to that of Corollary 8.28 of~\cite{brennan2020statistical} but it is given below for completeness.
\begin{lemma} \label{lem:SDA-bound}
    Let $0<c<1/2$ and $d,m \in \Z_+$ such that $d = m^{\Omega(1/c)}$. Consider the hypothesis testing problem of Problem~\ref{prob:non-gaussian-component-testing} where for every $y \in \R$ the distribution $A_y$ matches the first $m$ moments with $\cN(0,1)$ and $\E_{y \sim R(y)}[\chi^2(A_y,\mathcal{N}(0,1))] < \infty$. Then, for any $q \geq 1$,
    \begin{align*}
            \mathrm{SDA}\left( \mathcal{D}, \frac{\Omega(d)^{(m+1)(1/2-c)}}{\E_{y \sim R(y)}[\chi^2(A_y,\mathcal{N}(0,1))]\left(\frac{q^2}{2^{\Omega(d^{c/2})}} + 1\right)}  \right) \geq q \;.
    \end{align*}
\end{lemma}

    \begin{proof}
    The first part is to calculate the correlation of the set $\cS$ exactly as we did in the proof of Theorem~\ref{thm:main-formal}. By Lemma~\ref{lem:orthogonal_vectors}, Lemma~\ref{lem:lemma-3-4-DKS17} and Lemma~\ref{lem:form_of_conditional_distr} we know that the set $\mathcal{S}$ is $(\gamma,\beta)$-correlated with $\gamma = \Omega(d)^{-(m+1)(1/2-c)} \E_{y \sim R(y)}[\chi^2(A_y,\mathcal{N}(0,1))]$ and $\beta = \E_{y \sim R(y)}[\chi^2(A_y,\mathcal{N}(0,1))]$. 
    
    We next calculate the SDA according to Definition~\ref{def:SDA}. We denote by $\bar{E}_{v}$ the ratios of the density of $E_v$ to the density of $R$. Note that the quantity  $\langle \bar{E}_{u},\bar{E}_{,v} \rangle - 1$ used there is equal to $\langle \bar{E}_{u} - 1,\bar{E}_{v}  - 1\rangle$. Let $\cE$ be an event that has $\pr_{u,v \sim \mu}[\cE] \geq 1/q^2$. For $d$ sufficiently large we have that
    \begin{align*}
        \E_{u,v \sim \mu} [| \langle \bar{E}_{u} ,\bar{E}_{v} \rangle - 1 | \cE ] &\leq \min\left( 1,\frac{1}{|\mathcal{S}_{}| \pr[\cE]} \right) \E_{y \sim R(y)}[\chi^2(A_y,\mathcal{N}(0,1))] \\
        &+\max\left( 0,1-\frac{1}{|\mathcal{S}_{}| \pr[\cE]} \right) \frac{\E_{y \sim R(y)}[\chi^2(A_y,\mathcal{N}(0,1))]}{\Omega(d)^{(m+1)(1/2-c)}} \\
         &\leq \E_{y \sim R(y)}[\chi^2(A_y,\mathcal{N}(0,1))] \left( \frac{q^2}{2^{\Omega(d^c)}}  + \frac{1}{\Omega(d)^{(m+1)(1/2-c)}}\right) \\
         &= \E_{y \sim R(y)}[\chi^2(A_y,\mathcal{N}(0,1))] \frac{q^2 \Omega(d)^{(m+1)(1/2-c)} + 2^{\Omega(d^c)}   }{2^{\Omega(d^c)} \Omega(d)^{(m+1)(1/2-c)}}\\
         &= \E_{y \sim R(y)}[\chi^2(A_y,\mathcal{N}(0,1))] \left(  \frac{\Omega(d)^{(m+1)(1/2-c)}}{q^2 \Omega(d)^{(m+1)(1/2-c)} /2^{\Omega(d^c)} + 1} \right)^{-1} \\
         &= \E_{y \sim R(y)}[\chi^2(A_y,\mathcal{N}(0,1))] \left( \frac{\Omega(d)^{(m+1)(1/2-c)}}{q^2 /2^{\Omega(d^{c/2})} + 1} \right)^{-1} 
        \;,
    \end{align*}
     where the first inequality uses that $\pr[u=v | \cE] = \pr[u=v , \cE]/\pr[\cE]$ and bounds the numerator in two different ways: $\pr[u=v , \cE]/\pr[\cE] \leq \pr[u=v ]/\pr[\cE] = 1/(|\mathcal{D}_{}| \pr[\cE])$ and $\pr[u=v , \cE]/\pr[\cE] \leq \pr[\cE]/\pr[\cE] =1$.
\end{proof}    
We note that the lemma above and Theorem~\ref{thm:Sq_hardness_from_SDA} show SQ hardness of Problem~\ref{prob:non-gaussian-component-testing}. In the remainder of this section, we will apply these results to Problem~\ref{prob:hard-hypothesis-testing}.

\begin{corollary}
\label{cor:sda-hyp-testing}
    Let $0<c<1/2$, $m \in \Z_+$ with $m \leq c_1/\sqrt{\alpha}$ for some sufficiently small constant $c_1>0$ and $d = m^{\Omega(1/c)}$. Consider the hypothesis testing problem of Problem~\ref{prob:hard-hypothesis-testing}. Then, for any $k<{d^{c/4}}$:
    \begin{align*}
            \mathrm{SDA}\left( \mathcal{D}, \frac{\Omega(d)^{(2m+1)(1/2-c)}}{e^{O(m)}/(1-\rho^2)}  \right) \geq 100^k \;.
    \end{align*}
\end{corollary}
\begin{proof}
    We note that Problem~\ref{prob:hard-hypothesis-testing} is a special case of Problem~\ref{prob:non-gaussian-component-testing} (see Fact~\ref{fact:inliers_conditional_dist} and Lemma~\ref{lem:check_properties} which show that the conditional distributions are of the form $P_{A_y,v}$).  In Lemma~\ref{lem:SDA-bound} we use $q = \sqrt{2^{\Omega(d^{c/2})} (n/n')}$ with $n'=n  = \frac{\Omega(d)^{(2m+1)(1/2-c)}}{\E_{y \sim R(y)}[\chi^2(A_y,\mathcal{N}(0,1))]}$ to get that $\mathrm{SDA}(\cD, n) > 100^k$ for $k<{d^{c/4}}$. The first part of Lemma~\ref{lem:existence_of_A_y} states that the distributions $A_y$'s match the first $2m$ moments with $\cN(0,1)$ for $m\leq c_1/\sqrt{\alpha}$ and the second part implies that $\E_{y \sim R(y)}[\chi^2(A_y,\mathcal{N}(0,1))] = O(e^{m})/(1-\rho^2)$. This completes the proof.
\end{proof}

We conclude by noting the hardness of Problem~\ref{prob:non-gaussian-component-testing} and thus Problem~\ref{prob:hard-hypothesis-testing} in the SQ model. The proof of Theorem~\ref{thm:sq-hardness-hypothesis-testing} follows from Corollary~\ref{cor:sda-hyp-testing} and Theorem~\ref{thm:Sq_hardness_from_SDA}.

    \subsection{Reduction of Hypothesis Testing to List-Decodable Linear Regression}
\label{sec:reduction_from_regression}
    We now show that any list-decoding algorithm for robust linear regression can be efficiently used to solve Problem~\ref{prob:hard-hypothesis-testing}, that is, hypothesis testing efficiently reduces to list-decodable estimation.
    For a list $\cL$ and $i \in [|\cL|]$, we use $\cL(i)$ to denote the $i$-th element of $\cL$. 
    \begin{lemma} \label{lem:reduction}
        Let $d \in \Z_+$ with $d = 2^{\Omega(1/(1/2-c))}$. Consider the $(1{-\alpha})$-corrupted linear regression model of Definition~\ref{def:glr} with $\beta= \rho v$ for $v \in \cS^{d-1}$,
         $\rho \in (0,1)$, $\sigma^2 = 1-\rho^2$. 
         There exists an algorithm \textsc{List\_Regression\_To\_Testing} that, given a 
        list-decoding algorithm $\cA$ with the guarantee of returning a list $\cL$ of candidate vectors such that for some $i \in \{1,\ldots, |\cL| \}$, $\|\cL(i) -\beta \|_2 \leq \rho/4$, solves the hypothesis testing Problem~\ref{prob:hard-hypothesis-testing}
        with probability at least $1 - |\cL|^2 e^{-\Omega(d^{2c})}$. The running time of this reduction is quadratic in $|\cL|$.
    \end{lemma}
    \begin{proof}
        The reduction is described in Algorithm~\ref{alg:reduction}.
        \begin{algorithm}[ht]  \label{alg:reduction}
            \caption{Reduction from Hypothesis Testing to List-Decodable Linear Regression.} 
            \begin{algorithmic}[1]  
              \Statex  
              $\cA(\rho,(X_1,y_1),\ldots,(X_{n},y_{n}))$: List-decoding algorithm returning a list $L$ such that $\|\cL(i) - \beta\|_2 \leq \rho/4$ for some $i \in \{1,\ldots, |\cL| \}$.
              \Function{List\_Regression\_To\_Testing}{$\rho, (X_1,y_1),\ldots,(X_{2n},y_{2n})$} 
              \State Split dataset into two equally sized parts $\{(X_i,y_i)\}_{i=1}^n,\{(X'_i,y'_i)\}_{i=1}^n$.
              \State Let $A$ be a random rotation matrix independent of data. 
              \State $\cL_1  \gets \cA(\rho,(X_1,y_1),\ldots,(X_n,y_n))$.
              \State $\cL_2 \gets \cA(\rho,(A X'_{1},y'_{1}),\ldots,(A X'_{2n},y'_{n}))$.
              \For{$i \gets 1$ to $|\cL_1|$ }
                \For{$j \gets 1$ to $|\cL_2|$ }
                    \If{$\|\cL_1(i)\|_2,\| \cL_2(j)\|_2 \in [3\rho/4,5\rho/4]$ and $\|\cL_1(i) -  A^T\cL_2(j) \|_2 \leq \rho/2$}
                        \State \Return $H_1$ %
                    \EndIf
                \EndFor
              \EndFor
              \State \Return $H_0$
              \EndFunction  
            \end{algorithmic}  
          \end{algorithm}
    To see correctness, first assume that the alternative hypothesis holds. We note that the rotated points $(A X'_1,y'_1),\ldots,(A X'_n,y'_n)$ come from the Gaussian linear regression model of Definition~\ref{def:glr} having $\beta'=A\beta$ as the regressor. Thus $\cA$ finds lists $\cL_1,\cL_2$ such that there exist $i^* \in \{1,\ldots, |\cL_1|\}$  with $\| \cL_1(i^*) -  \beta\|_2 \leq \rho/4$ and  $j^* \in \{1,\ldots, |\cL_2|\}$ with $\|A^T\cL_2(j^*) -  \beta\|_2 \leq \rho/4$, where we use that $A^TA=I$. Moreover, since we are considering the regression model with $\|\beta\|_2 =\rho$, $\cL_1(i^*)$ and $A^T \cL_2(j^*)$ must have norms belonging in $[3 \rho /4, 5\rho/4]$. By the triangle inequality we get that $\|\cL_1(i^*) - A^T\cL_2(j^*)\|_2 \leq \rho/2$ and thus the algorithm correctly outputs $H_1$.

    Now assume that the null hypothesis holds, where the marginal on points is $\cN(0,I_d)$ and labels are independently distributed as $\cN(0,1/\alpha)$. Fix a pair $i \in [|\cL_1|]$, $j \in [|\cL_2|]$ for which $\|\cL_1(i)\|_2,\| \cL_2(j)\|_2 \in [3\rho/4,5\rho/4]$. Note that, by rotation invariance of the standard Gaussian distribution and the independence between covariates and response under the null distribution, the input $\{(AX'_i,y'_i)\}_{i=1}^n$ for the second execution of the list-decoding algorithm is independent of $A$. Thus the list $\cL_2$ is independent of $A$ (and also independent of $\cL_1$). Thus, $A^T\cL_2(j)$ is a random vector selected uniformly from the sphere of radius $\|\cL_2(j)\|_2$
    and independently of $\cL_1(i)$. Recall that two random vectors are almost orthogonal with high probability.
    \begin{lemma}[see, e.g., \cite{CaiFanJiang}]
        Let $\theta$ be the angle between two random unit vectors uniformly distributed over $\cS^{d-1}$. Then we have that $\pr[|cos \theta | \geq \Omega(d^{c-1/2})] \leq e^{-\Omega(d^{2c})}$ for any $0<c<1/2$.
    \end{lemma}
    Taking a union bound over the $|\cL_1|\cdot |\cL_2|$ possible pairs of candidate vectors, we have that with probability at least $1-|\cL_1|\cdot |\cL_2| e^{-\Omega(d^{2c})}$, for all $i \in [|\cL_1|], j \in[|\cL_2|]$ we have that
    \begin{align*}
        \|\cL_1(i) - A^T\cL_2(j)\|_2 &= \sqrt{\|\cL_1(i)\|^2_2 + \|A^T \cL_2(j)\|^2_2 - 2(\cL_1(i))^T ( A^T\cL_2(j))} \\
        &\geq \sqrt{ 2(3\rho/4)^2(1 - \Omega(d^{c-1/2}))}
        > \rho \;,
    \end{align*}
    where in the last inequality we used that $d = 2^{\Omega({1/(1/2-c)})}$. This concludes correctness for the case of the null hypothesis.
    \end{proof} 

We note that the Algorithm~\ref{alg:reduction} can be implemented in both of the models of computation that we consider: SQ model and low-degree polynomial test~(Section~\ref{sec:hardness_again_low_degree_poly}). For the SQ model, we can simulate the queries on the rotated $X$ by modifying the queries to explicitly perform the rotation on $X$ by a matrix $A$.
For the low-degree polynomial model, Remark~\ref{remark:test-in-low-degree-poly} shows that this reduction can be implemented as a low-degree polynomial algorithm.

\section{Hardness Against Low-Degree Polynomial Algorithms}
\label{sec:hardness_again_low_degree_poly}
In this section, we recall the recently established connection between the statistical query framework 
and low-degree polynomials, shown in~\cite{brennan2020statistical}, and deduce 
hardness results in the latter model. Section~\ref{sec:appendix-low-degree-basics} and Section~\ref{appendix:hardness-of-hypothesis-testing} are dedicated to the hypothesis problem. In Section~\ref{sec:appendix-reduction-in-low-degree}, 
we show that the reduction of Section~\ref{sec:reduction_from_regression} can be expressed as a low-degree polynomial test.

\subsection{Preliminaries: Low-Degree Method} \label{sec:appendix-low-degree-basics}

We begin by recording the necessary notation, definitions, and facts. This section mostly follows~\cite{brennan2020statistical}.
\paragraph{Notation} For a distribution $D$, we denote by $D^{\otimes n}$ the joint distribution of $n$ independent samples from $D$. For $f:\R \to \R$, $g:\R \to \R$ and a distribution $D$, we define the inner product $\langle f,g \rangle_D = \E_{X \sim D}[f(X)g(X)]$ and the norm $\snorm{D}{f} = \sqrt{\langle f,f \rangle_D }$.  We will omit the subscripts when they are clear from the context.

\paragraph{Low-Degree Polynomials} A function $f : \R^a \to \R^b$ is a polynomial of degree at most $k$ if it can be written in the   
form 
\begin{align*}
    f(x) = (f_1(x), f_2(x), \ldots, f_b(x) )\;,
\end{align*}
where each $f_i : \R^a \to \R$ is a polynomial of degree at most $k$. We allow polynomials to have random coefficients as long as they are independent of the input $x$. When considering \emph{list-decodable estimation} problems, an algorithm in this model of computation is a polynomial $f: \R^{d_1 \times n} \to \R^{d_2 \times \ell}$, where $d_1$ is the dimension of each sample, $n$ is the number of samples, $d_2$ is the dimension of the output hypotheses, and $\ell$ is the number of hypotheses returned. On the other hand, \cite{brennan2020statistical} focuses on \emph{binary hypothesis testing} problems defined in Definition~\ref{def:hypothesis-testing-general}. 

A degree-$k$ polynomial test for Definition~\ref{def:hypothesis-testing-general} is a degree-$k$ polynomial $f : \R^{d \times n} \to \R$ and a threshold $t \in \R$. The corresponding algorithm consists of evaluating $f$ on the input $x_1, \ldots, x_n$ and returning $H_0$ if and only if $f(x_1, \ldots, x_n) > t$.
\begin{definition}[$n$-sample $\eps$-good distinguisher]
    We say that the polynomial  $p : \R^{d\times n} \to \R$ is an $n$-sample $\eps$-distinguisher for the hypothesis testing problem in Definition~\ref{def:hypothesis-testing-general} if $\abs{\E_{X \sim D_0^{ \otimes n }} [p(X)]  - \E_{u \sim \mu} \E_{X \sim D_{u}^{\otimes n}} [p(X)] } \geq \eps \sqrt{ \var_{X \sim D_0^{\otimes n}} [p(X)]}$. We call $\eps$ the \emph{advantage} of the distinguisher.
\end{definition}
Let $\mathcal{C}$ be the linear space of polynomials with  degree at most $k$. The best possible advantage is given by the \emph{low-degree likelihood ratio}%
\begin{equation*}
    \max_{\substack{p \in \cC \\  \E_{X \sim D_0^{\otimes n}}[p^2(X)] \leq 1}}  \abs[\Big]{\E_{u \sim \mu} \E_{X \sim D_{u}^{\otimes n}} [p(X)] - \E_{X\sim D_0^{ \otimes n }} [p(X)]} 
    = \snorm{D_0^{ \otimes n }}{\E_{u \sim \mu}\left[(\bar{D}_u^{\otimes n})^{\leq k}\right]  -  1}\;,
\end{equation*}
where we denote $\bar{D}_u = D_u/D_0$ and the notation $f^{\leq k}$ denotes the orthogonal projection of $f$ to $\cC$.%

Another notation we will use regarding a finer notion of degrees is the following:  We say that the polynomial $f(x_1,\ldots, x_n) : \R^{d \times n} \to \R$ has \emph{samplewise degree} $(r,k)$ if it is a polynomial, where each monomial uses at most $k$ different samples from $x_1,\ldots, x_n$ and uses degree at most $d$ for each of them.  In analogy to what was stated for the best degree-$k$ distinguisher, the best distinguisher of samplewise degree $(r,k)$-achieves advantage $\snorm{D_0^{ \otimes n }}{\E_{u \sim \mu} [(\bar{D}_u^{\otimes n})^{\leq r,k}] - 1}$ the notation $f^{\leq r,k}$
now means the orthogonal projection of $f$ to the space of all samplewise degree-$(r,k)$ polynomials with unit norm.

\subsection{Hardness of Hypothesis Testing Against Low-Degree Polynomials} \label{appendix:hardness-of-hypothesis-testing}

In this section, we show the following result:
\begin{theorem}\label{thm:hypothesis-testing-hardness}
        Let $0<c<1/2$ and $m\in \Z_+$ with $m \leq c_1/\sqrt{\alpha}$ for some sufficiently small constant $c_1>0$. Consider the hypothesis testing problem of Problem~\ref{prob:hard-hypothesis-testing}. For $d \in \Z_+$ with  $d = m^{\Omega(1/c)}$, any $n \leq \Omega(d)^{(2m+1)(1/2-c)}e^{-O(m) }(1-\rho^2)$ and any even integer $k < d^{c/4}$, we have that
        \begin{align*}
            \snorm{R^{ \otimes n }}{\E_{u \sim \mu} \left[(\bar{E}_u^{\otimes n})^{\leq \infty,\Omega(k)}\right] - 1}^2 \leq 1\;.
        \end{align*}
    \end{theorem}
We prove Theorem~\ref{thm:hypothesis-testing-hardness} by using the lower bound on SDA in Corollary~\ref{cor:sda-hyp-testing} and the relation between SDA and low-degree polynomials established in \cite{brennan2020statistical}.
In~\cite{brennan2020statistical}, the following relation between $\mathrm{SDA}$ and low-degree likelihood ratio is established. 
\begin{theorem}[Theorem 4.1 of~\cite{brennan2020statistical}] \label{thm:sdaldlr}
    Let $\mathcal{D}$ be a hypothesis testing problem on $\R^d$ with respect to null hypothesis $D_0$. Let $n,k \in \mathbb{N}$ with $k$ even. Suppose that for all $0\leq n' \leq n$, $\mathrm{SDA}(\mathcal{S},n') \geq 100^k(n/n')^k$. Then, for all $r$, $\snorm{D_0^{ \otimes n }}{\E_{u \sim \mu} \left[ (\bar{D}_u^{\otimes n})^{\leq r,\Omega(k)} \right] -1}^2 \leq 1$.
\end{theorem}

    We first apply Theorem~\ref{thm:sdaldlr} to the more general Problem~\ref{prob:non-gaussian-component-testing}. In Lemma~\ref{lem:SDA-bound} we set $n = \frac{\Omega(d)^{(m+1)(1/2-c)}}{\E_{y \sim R(y)}[\chi^2(A_y,\mathcal{N}(0,1))]}$ and $q = \sqrt{2^{\Omega(d^{c/2})} (n/n')}$. Then, $\mathrm{SDA}(\mathcal{S},n') \geq  \sqrt{2^{\Omega(d^{c/2})} (n/n')} \geq (100n/n')^k$ for $k < d^{c/4}$. Thus, we have shown the following.
    \begin{corollary}\label{cor:low-deg-hardness-general-problem} 
        Let $0<c<1/2$ and the hypothesis testing problem of Problem~\ref{prob:non-gaussian-component-testing} where for every $y\in R$ the distribution $A_y$ matches the first $m$  moments with $\cN(0,1)$. For any $d \in \Z_+$ with $d = m^{\Omega(1/c)}$, any $n \leq \Omega(d)^{(m+1)(1/2-c)}/\E_{y \sim R(y)}[\chi^2(A_y,\mathcal{N}(0,1))]$ and any even integer $k < d^{c/4}$, we have that
        \begin{align*}
            \snorm{R^{ \otimes n }}{\E_{u \sim \mu} \left[(\bar{D}_{u}^{\otimes n})^{\leq \infty,\Omega(k)}\right] - 1}^2 \leq 1\;.
        \end{align*}
    \end{corollary}

     \begin{proof}[Proof of Theorem~\ref{thm:hypothesis-testing-hardness}]
      We now apply the Corollary~\ref{cor:low-deg-hardness-general-problem} to Problem~\ref{prob:hard-hypothesis-testing}, which is a special case of Problem~\ref{prob:non-gaussian-component-testing}. The first part of Lemma~\ref{lem:existence_of_A_y} states that the distributions $A_y$'s match the first $2m$ moments with $\cN(0,1)$ for $m\leq c_1/\sqrt{\alpha}$ and the second part implies that $\E_{y \sim R(y)}[\chi^2(A_y,\mathcal{N}(0,1))] = O(e^{m})/(1-\rho^2)$. An application of Corollary~\ref{cor:low-deg-hardness-general-problem} completes the proof.
        \end{proof}

\subsection{Low-Degree Polynomial Reduction to List-Decodable Regression} \label{sec:appendix-reduction-in-low-degree}

    \begin{remark}
    \label{remark:test-in-low-degree-poly}
        We note that the reduction of Lemma~\ref{lem:reduction} is an algorithm that can be expressed in the low-degree polynomials model. The modification of the algorithm is the following: First note that  the $\ell_2$-norm of a vector is indeed a polynomial of degree two in each coordinate. Second, one can check whether there exists a pair $i \in [|\cL_1|], j \in  [|\cL_2|]$ with $\|\cL_1(i) \|_2,\|\cL_2(j) \|_2 \in [3\rho/4,5\rho/4]$ for which $\|\cL_1(i) - A^T \cL_2(j) \|_2 \leq \rho/2$ using the condition 
        \begin{align*}
           \sum_{i \in 1}^{|\cL_1|} \sum_{j \in 1}^{|\cL_2|} \mathbf{1}(\|\cL_1(i) \|_2^2\geq (3\rho/4)^2) \cdot \mathbf{1}(\|A^T\cL_2(j) \|_2^2\leq(5\rho/4)^2)  \cdot \mathbf{1}(\|\cL_1(i) - A^T \cL_2(j) \|_2^2 \leq \rho^2/4) = 0 \;,
        \end{align*}
        and use a polynomial approximation for the step function in order to express each term as a polynomial. The degree needed for a uniform $\eps$-approximation has been well-studied~\cite{ganzburg2008limit,ganzburg2002limit,eremenko2007uniform}.
        \begin{lemma}[\cite{eremenko2007uniform}]
            Let $f: \R \to \R$ be the step function defined as $f(x) = 1$ for all $x \geq 0$ and $f(x)=0$ otherwise. The minimum $k \in \Z_+$ for which there exists a degree-$k$ polynomial $p : \R \to \R$ such that $\max_{x \in [-1,1]} | f(x) - p(x) | \leq \eps$ is $k = \Theta(1/\eps^2)$.
        \end{lemma}
        For our purpose, it suffices to approximate the step function up to error $\eps = \Theta(1/(|\cL_1|\cdot |\cL_2|))$, thus the resulting polynomial test has degree $\Theta(|\cL_1|^2\cdot |\cL_2|^2)$.
    \end{remark}

\clearpage

\bibliographystyle{alpha}
\bibliography{allrefs}

\newpage

\appendix

\section*{Appendix}

    \section{Additional Technical Facts} \label{sec:appendix-extras}
    
    Our bounds in Lemma~\ref{lem:existence_of_A_y} required the fact below. Here we provide its proof for completeness.
    \begin{fact} 
		For any one-dimensional distribution $P$ that matches the first $m$ moments with $\cN(0,1)$ and has $\chi^2(P,\cN(0,1)) < \infty$ the following identity is true
		\begin{align*}
			\chi^2(P,\cN(0,1)) = \sum_{i=m+1}^\infty \left(\E_{X \sim P}[h_i(X)]\right)^2 \;.
		\end{align*}
	\end{fact}
	\begin{proof}
	    Let $\phi$ denote the pdf of the standard one-dimensional Gaussian. For this proof, we use a slightly different definition of the space $L^2(\R,\cN(0,1))$. We define it as the space of functions for which $\int_{\R}f^2(x)/\phi(x)\dx < \infty$ with the inner product $\langle f,g \rangle := \int_{\R}f(x)g(x)/\phi(x) \dx$ (note the similarity with the definition of $\chi^2$-divergence). The \emph{Hermite functions} (or often called \emph{Hermite-Gauss functions}) $h_i(x) \phi(x)$ for  $i=0,1,\ldots$ form a complete orthonormal basis of the space $L^2(\R,\cN(0,1))$
	    with respect to that inner product.
	    It is easy to check that this statement is equivalent to the statement that Hermite polynomials $\{h_i\}_{\N}$ form a complete orthonormal basis of the space of all functions $f:\R \to \R$ for which $\E_{x \sim \cN(0,1)}[f^2(x)] < \infty$ (i.e., our old definition of $L^2(\R,\cN(0,1))$).
	    Since $\chi^2(P,\cN(0,1))<\infty$ we have $P \in L^2(\R,\cN(0,1))$ and thus  we can write $P(x) = \sum_{i=0}^\infty a_i h_i(x) \phi(x)$, where $a_i = \E_{X \sim P}[h_i(X)]$. Using the fact that $P$ agrees with the first $m$ moments of $\cN(0,1)$ and the property of Hermite polynomials $\E_{X \sim \cN(0,1)}[h_i(X)] = \mathbf{1}(i=0)$ we get that $a_0 = \E_{X \sim \cN(0,1)}[h_0(X)] = 1$ and $a_i = \E_{X \sim \cN(0,1)}[h_i(X)] = 0$ for $0<i\leq m$. Thus
	    \begin{align*}
	        P(x) = \phi(x) + \sum_{i=m+1}^\infty a_i h_i(x) \phi(x) \;.
	    \end{align*}
	    The $\chi^2$-divergence can then be written as
	    \begin{align*}
	        \chi^2(P,\cN(0,1)) = \int_{\R} \frac{(P(x) - \phi(x))^2}{\phi(x)} \dx 
	        = \int_{\R} \frac{1}{\phi(x)}\left( \sum_{i=m+1}^\infty a_i h_i(x) \phi(x)\right)^2 \dx 
	        = \sum_{i=m+1}^\infty a_i^2 \;,
	    \end{align*}
	    where the last part uses orthonormality of the functions $h_i(x) \phi(x)$.
	\end{proof}
	
	We now turn to Claim~\ref{claim:finite-chi-square} which is restated below.
	\begin{claim} \label{claim:finite-chi-square-restated} 
	If $P = \sum_{i=1}^k \lambda_i N(\mu_i,\sigma_i^2)$ with $\mu_i \in \R$, $\sigma_i < \sqrt{2}$ and $\lambda_i \geq 0$ such that $\sum_{i=1}^k\lambda_i = 1$, we have that $\chi^2(P,\cN(0,1)) < \infty$.
	\end{claim}
	For that we need the following two facts about $\chi^2$-distance between Gaussians. Their proofs can be done by direct calculations.
	\begin{fact} \label{fact:chi_square_for_gaussians1}
	        Let $k \in \Z_+$, distributions $P_i$ and $\lambda_i\geq 0$, for $i \in [k]$ such that  $\sum_{i=1}^k \lambda_i=1$. We have that $\chi^2\left(\sum_{i=1}^k \lambda_i P_i,D  \right)= \sum_{i=1}^k \sum_{j=1}^k \lambda_i \lambda_i \chi_D(P_i,P_j)$.
	\end{fact}
	\begin{proof}
	    \begin{align*}
	        \chi^2\left(\sum_{i=1}^k \lambda_i P_i,D  \right) + 1 &= \int_{\R} \left( \sum_{i=1}^k \lambda_i P_i(x)\right)^2/D(x) \dx 
	        = \sum_{i=1}^k \sum_{j=1}^k \lambda_i \lambda_j \int_\R  P_i(x)P_j(x)/D(x)\dx \\
	        &= \sum_{i=1}^k \sum_{j=1}^k \lambda_i \lambda_j \left(  \chi_D(P_i,P_j) + 1\right) 
	        =  \sum_{i=1}^k \sum_{j=1}^k \lambda_i \lambda_j  \chi_D(P_i,P_j) + \left(\sum_{i=1}^k \lambda_i\right)^2\\
	        &=  \sum_{i=1}^k \sum_{j=1}^k \lambda_i \lambda_j  \chi_D(P_i,P_j) + 1 \;.
	    \end{align*}
	\end{proof}
	\begin{fact}\label{fact:chi_square_for_gaussians3}
	\begin{align*}
	    \chi_{\cN(0,1)}\left(\cN(\mu_1,\sigma_1^2), \cN(\mu_2,\sigma_2^2)\right) = \frac{\exp\left(-\frac{\mu_1^2(\sigma_2^2-1) +2\mu_1 \mu_2 + \mu_2^2(\sigma_1^2-1)}{2\sigma_1^2(\sigma_2^2-1) - 2\sigma_2^2} \right)}{\sqrt{\sigma_1^2 + \sigma_2^2 - \sigma_1^2\sigma_2^2}} - 1\;.
	\end{align*}
	\end{fact}
	The proof of Claim~\ref{claim:finite-chi-square-restated} then consists of applying Fact~\ref{fact:chi_square_for_gaussians1} and using Fact~\ref{fact:chi_square_for_gaussians3}  for each one of the generated terms.

\end{document}